\documentclass[orivec]{llncs}

\usepackage{amsmath, amssymb}
\usepackage{amsxtra}
\usepackage{stmaryrd}
\usepackage{mathpartir}
\usepackage{url}
\usepackage{enumerate}
\usepackage{latexsym, graphics, graphicx}
\usepackage{float}
\usepackage{xspace}
\usepackage{listings}
\usepackage{newfloat}
\DeclareFloatingEnvironment[fileext=lop]{listing}
\usepackage{color}
\usepackage{array}
\usepackage{booktabs}
\usepackage{paralist}
\usepackage{tikz}
\usetikzlibrary{automata}

\newcommand{\bsqcup}{{\textstyle\bigsqcup}}

\newcommand{\set}[1]{\ensuremath{\{{#1}\}}}
\newcommand{\setx}[2]{\ensuremath{\{#1 \mid #2\}}}
\newcommand{\ident}[1]{\ensuremath{\mathit{#1}}}
\newcommand{\Powerset}[1]{\mathcal{P}({#1})}
\newcommand{\PowersetFin}[1]{\mathcal{P}^{\mathit{fin}}({#1})}

\newcommand{\FnTo}{\ensuremath{\rightarrow}}
\newcommand{\FnToFin}{\ensuremath{\rightharpoonup}}
\newcommand{\domof}{\ensuremath{dom}}
\newcommand{\dom}{\operatorname{dom}}

\newcommand{\seq}[1]{\ensuremath{\bar{#1}}}
\newcommand{\seqvar}[1]{\ensuremath{\langle{#1}\rangle}}

\newcommand{\Rule}[3]{\ensuremath{\inferrule*[left={#1}]{#2}{#3}}}
\newcommand{\rulename}[1]{\textsc{#1}}

\newcommand{\Monoid}{\mathit{Mon}}
\newcommand{\wordoflit}{\ensuremath{\mathsf{lit2word}}}
\newcommand{\monneutral}{[\varepsilon]}
\newcommand{\strcl}[1]{\texttt{#1}}

\newcommand{\Allowed}{\mathsf{Allowed}}

\newcommand{\Strings}{\mathit{Str}}

\newcommand{\FJEUS}{FJEUCS\xspace}

\newcommand{\syntax}[1]{\ensuremath{\mathtt{#1}}}
\newcommand{\Null}{\syntax{null}}
\newcommand{\New}[1]{\syntax{new}\;{#1}}
\newcommand{\Let}[3]{\syntax{let}\;{#1}={#2}\;\syntax{in}\;{#3}}
\newcommand{\IfEqual}[4]{\syntax{if}\;{#1}={#2}\;\syntax{then}\;{#3}\;\syntax{else}\;{#4}}
\newcommand{\Assign}[2]{{#1}\,\syntax{:=}\,{#2}}
\newcommand{\Cast}[2]{(#2)\,#1}

\newcommand{\Str}{\ensuremath{str}}
\newcommand{\LExpr}[2]{[#1]^{#2}}

\newcommand{\this}{\mathit{this}}
\newcommand{\KString}{\syntax{String}}
\newcommand{\KObject}{\syntax{Object}}

\newcommand{\NullType}{\syntax{NullType}}

\newcommand{\Vars}{\mathit{Var}}
\newcommand{\Classes}{\mathit{Cls}}
\newcommand{\Fields}{\mathit{Fld}}
\newcommand{\Methods}{\mathit{Mtd}}
\newcommand{\eMethods}{\mathit{Fn}}
\newcommand{\Exprs}{\mathit{Expr}}
\newcommand{\Sites}{\mathit{Pos}}

\newcommand{\fn}{\mathit{fn}} 
\newcommand{\ar}{\mathit{ar}} 

\newcommand{\subclassOf}{\ensuremath{\preceq}}
\newcommand{\fList}{\mathit{fields}}
\newcommand{\mList}{\mathit{methods}}
\newcommand{\mTable}{\mathit{mtable}}
\newcommand{\Prg}{\ensuremath{P}}

\newcommand{\NullVal}{\mathit{null}}
\newcommand{\Values}{\mathit{Val}}
\newcommand{\Locs}{\mathit{Loc}}
\newcommand{\Objs}{\mathit{Obj}}
\newcommand{\SObjs}{\mathit{SObj}}

\newcommand{\classOfVal}{\mathit{classOf}}

\newcommand{\FSem}[6]{\ensuremath{({#1},{#2}) \vdash {#3} \Downarrow
    {#4},{#5} \;\&\; {#6}}}
\newcommand{\sem}{\mathit{sem}}

\newcommand{\Regions}{\mathit{Reg}}
\newcommand{\Contexts}{\mathit{Ctx}}
\newcommand{\Types}{\mathit{Typ}}
\newcommand{\ATypes}{\mathit{ATyp}}
\newcommand{\TPool}{T} 

\newcommand{\subtypeOf}{\ensuremath{<:}}
\newcommand{\Effects}{\mathit{Eff}}

\newcommand{\msqsubseteq}{\sqsubseteq_{\mathsf{m}}}

\newcommand{\WellFormed}[3]{\ensuremath{{#1}\vdash{#2}:{#3}}}
\newcommand{\MethodTP}[3]{\ensuremath{({#1},{#2},{#3})}}

\newcommand{\SType}[5]{\ensuremath{{#1};{#2}\;\vdash_{\textsf{s}}\;{#3}:\,{#4}\;\&\;{#5}}}
\newcommand{\SemHeap}[2]{\ensuremath{{#1}\vdash{#2}}}
\newcommand{\SemHeapVal}[3]{\ensuremath{{#1}\vdash{#2}:{#3}}}
\newcommand{\HTextends}[2]{\ensuremath{{#1} \sqsupseteq {#2}}}
\newcommand{\WellTyped}[3]{\ensuremath{{#1}\vdash {#2}:{#3}}}

\newcommand{\FType}[4]{\ensuremath{{#1}\;\vdash_{\textsf{d}}\;{#2}:\,{#3}\;\&\;{#4}}}

\newcommand{\AType}[5]{\ensuremath{{#1};{#2}\;\vdash_{\textsf{a}}\;{#3}:\,{#4}\;\&\;{#5}}}

\newcommand{\Lattice}{\mathcal{L}}

\newcommand{\alg}{\mathcal{A}}

\newcommand{\te}{\mathsf{typeff}}
\newcommand{\tep}{\mathsf{typeff}'}

\newcommand{\atoms}{\operatorname{atoms}}

\lstdefinelanguage{myML}{
  backgroundcolor = \color{lightgray!20!},
  morekeywords={match,with,proc,assert,for,foreach,range,do,while,until,break,if,then,else,return,and,or,raise,false,true},
  columns=fullflexible,
  sensitive=true,
  commentstyle = \itshape, 
  morecomment={[l]//},
  mathescape=true,
  basicstyle=\small,
  identifierstyle={\ttfamily},
  stringstyle=\rmfamily,
  literate={<-}{$\leftarrow\ $}{2} {->}{$\rightarrow\ $}{2} {:=}{$\leftarrow\ $}{2}
}

\lstdefinelanguage{fjeus}{
  sensitive=true,
  morecomment=[l]{//},
  morestring=[b]",
  basicstyle=\normalsize,
  stringstyle=\rmfamily,
}

\newcommand{\ls}{\lstinline[language=fjeus,basicstyle=\normalsize]}

\newcommand{\mpar}[1]{}

\newcommand{\ignore}[1]{}

\begin{document}

\title{Enforcing Programming Guidelines with Region~Types and Effects\thanks{This research is funded by the German Research Foundation
  (DFG) under research grant 250888164 (GuideForce).}
}


\author{
 	Serdar Erbatur\inst{1}
 	\and
 	Martin Hofmann\inst{1}
 	\and
 	Eugen Z\u alinescu\inst{2}
}
\institute{
 	Ludwig-Maximilians-Universit\"{a}t,
 	Munich, Bavaria, Germany\\
 	\email{\{serdar.erbatur, hofmann\}@ifi.lmu.edu}
  \and
  Institut f\"ur Informatik, 
  Technische Universit\"at M\"unchen, Germany
  \email{eugen.zalinescu@in.tum.de}
}


\maketitle

\begin{abstract}
We present in this paper a new type and effect system for Java which
can be used to ensure adherence to guidelines for secure web
programming. The system is based on the region and effect system by
Beringer, Grabowski, and Hofmann. 
It improves upon it by being parametrized over an arbitrary guideline 
supplied in the form of
a finite monoid or automaton and a type annotation or mockup code for
external methods. Furthermore, we add a powerful type inference based on
precise interprocedural analysis and provide an implementation in the
Soot framework which has been tested on a number of benchmarks
including large parts of the Stanford SecuriBench.

\end{abstract}

\section{Introduction}
We present in this paper a new type and effect system for Java which
can be used to ensure adherence to guidelines for secure web
programming such as proper sanitization of externally supplied strings
or appropriate authorization prior to access to sensitive data.
Unlike its precursors, the system can be freely configured and in this
way it can guarantee adherence to a whole host of such guidelines.

The type system is based on the region and effect systems given
in~\cite{comlan,fast} but improves upon and extends them in a number
of ways. First, our system is parametrized by an arbitrary monoid
abstracting both string values and sequences of events such as writing
certain strings to files or invoking certain framework
methods.

Second, in~\cite{comlan} heuristic context information was used in two
places: first, in order to index several types for one and the same
method and thus to provide a limited amount of polymorphism and
secondly in order to determine regions for newly allocated objects. As
a side effect this provided for the first time a rigorous
type-theoretic underpinning for context-sensitive and points-to
analyses. The system presented here keeps heuristic, user-dependent
context information for the points-to part, i.e.~to determine regions
for newly allocated objects, but uses precise and fully automatic
interprocedural analysis for method~typings.

Third, we have implemented an automatic type inference for the
system using the Soot framework~\cite{lam2011soot}. This makes our
analysis applicable to actual Java source code via Soot's builtin
translation into control-flow graphs which are sufficiently close to
the academic language Featherweight Java~\cite{IgarashiPW01} extended with updates and
strings (FJEUS) on which our type system and its theoretical analysis
are based. This allowed us to test the type system on a number of
interesting benchmarks including those from the SecuriBench and thus
to practically demonstrate the assertion already made in \cite{comlan}
that region-based type and effect systems achieve the same accuracy as
state-of-the-art context-sensitive pointer and points-to analyses
\cite{ChristensenMS03,CregutA05,Lenherr08,Lhotak06}.

Being formulated as a type system, our analysis enjoys a very clear
semantic foundation in the form of a declarative type system which
consists of fairly simple and natural rules and is sound with respect
to a standard operational semantics. From there we move in three steps
to the actual implementation of type inference. Each step is again
sound with respect to the previous one. These steps are first a
semi-declarative type system in which regions for newly allocated
objects are selected deterministically as a function of program point
and a context abstraction chosen by the user. At this point we may
lose some precision but not soundness because the context abstraction
might have been ill-chosen (this also happens in the case of classical
points-to analysis). We then give an algorithmic type system whose
rules are syntax-directed and can be read as a logic program. This
latter system is sound \emph{and} complete with respect to the
semi-declarative system. The last step, finally, is the actual
implementation using Soot which replaces the logic program implicit in
the algorithmic type system with an actual fixpoint iteration both
over control flow graphs (intraprocedural part) and over method call
graphs (interprocedural part). Again, the algorithm underlying the
implementation is sound and complete with respect to the algorithmic
type system. We thus have semantic soundness of all stages and we can
clearly delineate where precision is lost: first, we have the
declarative type system which formalizes a certain degree of
abstraction, e.g., by treating conditionals as nondeterministic choice
(path-insensitivity)
and by presupposing a finite set of regions.
Second, we have the passage from the declarative to the
semi-declarative system. All other stages are precision-preserving.

Before we start let us illustrate the approach with two simple examples.

\begin{example}\label{ex:intro}
  Consider the following small Java program
  to be subjected to tainted\-ness
  analysis. User input (as obtained by \ls{getString}) is
  assumed to be tainted and should not be given to \ls{putString} as
  input without preprocessing.

\hspace*{2em}
\begin{minipage}{0.5\textwidth}
\begin{lstlisting}[language=java]
class C {
  main() {
    D f1 = new D(getString());
    D f2 = new D("test");
    putString(f2.s); }}
\end{lstlisting}
\end{minipage}
\begin{minipage}{0.45\textwidth}
\begin{lstlisting}[language=java]
class D {
  public String s;
  public D(String s) {
    this.s = s; }}
$\phantom{I}$
\end{lstlisting}
\end{minipage}
To ensure this, our type system will refine \ls{String}
into two types: \ls{String@user} containing ``tainted'' strings and
\ls{String@ok} which contains untainted strings such as literals
from the program text or results from trusted sanitization
functions. We also refine the class \ls{D} using two \emph{regions}
into the refined class types \ls{D@red} and \ls{D@green}. For each
of these we then have refined field and method typings: the~\ls{s}
field of \ls{D@red} objects is typed \ls{String@user} and so is
the parameter of their constructor. The class \ls{D@green} uses
\ls{String@ok} instead. If there are more fields and methods that we
want to differentiate we might need more regions than just two. With
these typings in place, we can then see that the program is correct as
follows: the variable \ls{f1} gets type \ls{D@red} whereas
\ls{f2} gets type \ls{D@green}. Thus, \ls{f2.s} has type
\ls{String@ok} and the external method call is permitted. We notice
that if we had allowed only one region rather than two, i.e. without regions, 
we would be forced to give field  would \ls{s} the type \ls{String@user} in
view of the assignment to~\ls{f1}. We thus see how the regions provide object sensitivity. If, on the
other hand, we had erroneously written \ls{putString(f1.s)} then,
since \ls{f1.s} has type \ls{String@user} a type error would have
resulted no matter how many regions we use. 
\end{example}

Consider now another example in which we need to take into account method 
effects, type casting, and library methods for which source code is not available.
\begin{example}\label{ex:intro2}
  Consider the following method. 
  \begin{lstlisting}[language=java]
void doGet(HttpServletRequest req, HttpServletResponse resp) throws IOException {
  String s1 = req.getParameter("name");
  LinkedList<String> list = new LinkedList<String>();
  list.addLast(s1);
  String s2 = (String) list.getLast();
  PrintWriter writer = resp.getWriter();
  writer.println(s2);                    /* BAD */ }
\end{lstlisting}
In the last line, a possibly tainted string is written unsanitized.
\end{example}

We use two conceptually different ways to handle external
methods whose code is not part of the program being analyzed but comes
from some ambient framework or library: (1) builtin methods
(like \ls{getString}) that always take and return
strings and are given semantically; (2) 
external methods
(like \ls{addLast}) that take and return values
other than strings and are defined with mockup~code.

For the builtin methods we provide their semantic behavior and typing by
hand also taking into account tags representing taintedness, the
action of sanitization functions, or similar. We also specify event
traces produced by them. To this end, we assume an alphabet to
represent tags and event traces, and use finite automata and monoids
to obtain a finite abstraction of the event traces produced by the
program. This eventually allows us to ensure that the set of traces
produced by a program will be accepted by a given policy automaton.

Defined external methods are altogether different.  We use
mockup code to represent their implementation, so as to obtain a model
of the part of the ambient framework that is relevant to the analysis.
Mockup code is thus analyzed along with the user code. 
For the given example, we use the following mockup code for linked
lists, where we want to impose the abstraction that all their
entries are treated the same, e.g., if one entry is tainted then all
entries count as tainted.

\noindent
\begin{minipage}{0.4\textwidth}
  \begin{lstlisting}[language=java]
class LinkedList { 
  Object o; 
  LinkedList(Object e) { o = e; } 
  \end{lstlisting}
\end{minipage}
\begin{minipage}{0.55\textwidth}
\begin{lstlisting}

  boolean addLast(Object e) { o = e; return true; } 
  Object getLast(int index) { return o; } ... }
\end{lstlisting}
\end{minipage}

\noindent
The following mockup code implements the classes
\ls{HttpServletRequest}, etc.~and their relevant methods in terms of
the builtin methods \ls{getString} and~\ls{putString}.
\begin{lstlisting}[language=java]
class HttpServletRequest { String getParameter(String s) { return getString(); } }
class HttpServletResponse { PrintWriter getWriter() { return new PrintWriter(); } }
class PrintWriter { void println(String s) { return putString(s); } }
\end{lstlisting}

An alternative to mockup code which we had also explored consists of
ascribing refined types to library methods which then have to be
justified manually. However, we found that the expressive power of
mockup code is no less than that of arbitrary types and often
intuitively clearer. This may be compared to the possibility of
approximating stateful behavior of an unknown component with a state
machine or alternatively (that would correspond to the use of manually
justified types) with a temporal logic formula.

Lastly, to handle casts in general we extended FJEUS in~\cite{comlan}
with appropriate semantics and refined type system. Regarding the 
application scope, we must add that our scope is not limited to 
guidelines for strings. Currently we are able to formalize guidelines 
that represent safety properties (see Appendix~\ref{app:auth} for
an application to a guideline for authorization) and further plan 
to extend our approach to liveness and fairness properties.

We stress that our type inference algorithm automatically finds and
checks all field and method typings. There is no need for the user to fill in
these typings. All a user needs to provide is: 
\begin{compactitem}
\item the policy or guideline to be checked; 
\item typings or mockup code for external framework methods; 
\item a context abstraction (usually taken from a few standard ones); 
\item the set of regions to be used (or their number).
\end{compactitem}

\section{Formalizing Programming Guidelines}
\label{sec:policy}

In order to formalize a given guideline to be enforced we select a
finite alphabet~$\Sigma$ to represent string \emph{tags},
i.e.~annotations that denote some extra information about strings,
like taintedness, input kind etc.  This same alphabet is also used to
represent \emph{events} generated during the program execution, like
writing to a file or the invocation of external methods.
One could use different alphabets for tags and events but for
simplicity we chose not to do so.

We also assume an infinite set~$\Strings$ of \emph{string literals},
and a function $\wordoflit:\Strings\to\Sigma^*$ that specifies the tag
word $w\in\Sigma^*$ for a given string literal. 

The operational semantics is then instrumented so that string literals
are always paired with their tag words in~$\Sigma^*$ and evaluation of
an expression always results in an event trace also represented as a
finite word over $\Sigma$.

Next, we require a finite monoid $\Monoid$ and a homomorphism
$[-]:\Sigma^*\rightarrow \Monoid$ providing a finite abstraction of
string tags and event traces. 
We let $\Effects := \Powerset{\Monoid}$ and use this abbreviation when
referring to event trace abstractions, in contrast to string tag
abstractions.
We single out a subset
$\Allowed \subseteq \Monoid$ consisting of the event traces that are
allowed under the guideline to be formalized.

Then, we need to specify a collection of builtin methods such as
\ls{getString} and \ls{putString} all of
which are static and only take string parameters and return a
string. We use $\fn$ to range over these builtin methods.

For each $n$-ary builtin method  $\fn$ we assume given a function
\[
\sem(\fn) : (\Strings\times \Sigma^*)^n \FnTo \mathcal{P}(\Strings\times \Sigma^* \times \Sigma^*).
\]
The intuition is that when
$(t,w,w')\in\sem(\fn)((t_1,w_1),\dots,(t_n,w_n))$ then this means that
the invocation of the builtin method $\fn$ on string arguments
$t_1,\dots,t_n$ tagged as $w_1,\dots,w_n$ can yield result string $t$
tagged as $w$ and leave the event trace $w'$.
The nondeterminism, i.e.~the fact that $\sem(\fn)$ returns a set of
results allows us for instance to represent user input,
incoming requests, or file contents.

Furthermore, we must specify a ``typing'' in the form of a function
\[
M(\fn) : \Monoid^n\rightarrow \mathcal{P}(\Monoid)\times\Effects
\]
such that whenever $t_1, \ldots, t_n\in\Strings$ and
$w_1,\dots,w_n\in\Sigma^*$ and $(t,w,w')\in$
$\sem(\fn)((t_1,w_1),\dots,(t_n,w_n))$, it holds that $[w]\in U$ and
$[w']\in U'$, where $(U,U') = M(\fn)([w_1], \dots, [w_n])$.

In summary, the formalization of a guideline thus comprises:
\begin{compactitem}
\item tag and event alphabet, finite monoid,\footnote{Alternatively,
    the monoid can be generated automatically from the policy
    automaton (or its complement).} allowed subset, and homomorphism,
\item semantic functions and typings for builtin methods,
\item mockup code for other framework components.
\end{compactitem}

\begin{example}\label{ex:guideline}
  Consider the program in Example~\ref{ex:intro}. We set
  $\Sigma=\set{\mathsf{user},\mathsf{ok}}$,
  $\Monoid=\set{\mathsf{T},\mathsf{U}}$ ($\mathsf{T}, \mathsf{U}$ standing for ``tainted''/``untainted'') with $\mathsf{U}\mathsf{U} = \mathsf{U}$ and
  $\mathsf{U}\mathsf{T} = \mathsf{T}\mathsf{U} = \mathsf{T}\mathsf{T} = \mathsf{T}$.
  Also, $[\mathsf{user}]=\mathsf{T}$, $[\mathsf{ok}]=\mathsf{U}$, and
  $\Allowed=\set{\mathsf{U}}$.
  The semantic function $\sem(\strcl{getString})$ returns
  nondeterministically a triple of the form
  $(t,\seqvar{\strcl{user}},\varepsilon)$ with $t$ a string, $\varepsilon$
  the empty trace, and $\seqvar{-}$ denoting sequences. 
  The semantic function $\sem(\strcl{putString})(t,w)$ returns
  the triple $(\mathsf{""},\varepsilon,\seqvar{[w]})$, with empty strings
  representing \ls{void}. 
  The typing of $\strcl{getString}$ is
  $M(\strcl{getString})() = (\set{[\mathsf{user}]},
  \set{[\varepsilon]}) = (\set{\mathsf{T}}, \set{\mathsf{U}})$.  (Note
  that $\mathsf{U}$ is the monoid's neutral element.)  The typing of
  $\strcl{putString}$ is
  $M(\strcl{putString})(u) = (\set{\mathsf{U}}, \set{u})$, for any $u\in\Monoid$.
\end{example}

In Appendix~\ref{app:taint} we give an example of a more elaborate
guideline for taintedness analysis, which also takes into account
sanitizing functions.
%

\section{Featherweight Java with Updates, Casts, and Strings}
\label{sec:fjeus}

\FJEUS is a formalized and downsized object-oriented language that
captures those aspects of Java that are interesting for our analysis:
objects with imperative field updates, and strings.  The language is
an extension of FJEUS~\cite{fast} with casts, which itself extends
FJEU~\cite{HofmannJ06} with strings, which itself extends
Featherweight Java~\cite{IgarashiPW01} with side effects on a
heap.

\subsection{Syntax}

The following table summarizes the (infinite) abstract identifier sets
in \FJEUS, the meta-variables we use to range over them, and the
syntax of expressions. Notice that we only allow variables in various
places, e.g., field access so as to simplify the metatheory. By using
let-bindings, the effect of nested expressions can be restored (let
normal form).
\begin{gather*}
\begin{array}{>{$}r<{$}rclc>{$}r<{$}rclc>{$}r<{$}rcl}
    variables: & x,y &\in& \Vars &\quad\quad\quad &  
       fields:    & f &\in& \Fields &\quad\quad\quad &
    string literals: & \Str & \in & \Strings\\
      classes:   & C,D &\in& \Classes && 
      methods:   & m &\in& \Methods &&
  builtin methods: & \fn &\in& \eMethods\\
\end{array}
\\
\begin{array}{rcl}
  \Exprs \ni e & ::= & x \mid \Let{x}{e_1}{e_2} \mid
            \IfEqual{x_1}{x_2}{e_1}{e_2} \mid \Null \mid \New C\\
         && \mid \Cast{e}{C} \mid x.f \mid \Assign{x_1.f}{x_2} \mid 
         x.m(\seq y) \mid \fn(\seq y) \mid
          "\Str" \mid x_1 + x_2 
\end{array}
\end{gather*}
We assume that $\Classes$ contains three distinguished elements,
namely $\KObject$, $\KString$, and $\NullType$.\footnote{In Java, the
  $\NullType$ is the type of the expression $\Null$, see
  \url{https://docs.oracle.com/javase/specs/jls/se7/html/jls-4.html\#jls-4.1}.}
In $\New C$ expressions, we require that $C\neq\KString$ and $C\neq\NullType$. 
Also, in $\Cast{e}{C}$ expressions, $C\neq\NullType$.

An \FJEUS program over the fixed set of builtin methods is defined by the following relations and functions.
\begin{displaymath}
  \begin{array}{>{$}r<{$}rcl}
    subclass relation: & 
      \prec & \in & \PowersetFin{\Classes \times \Classes}\\
    field list:  & 
      \fList & \in & \Classes \FnTo \PowersetFin\Fields\\
    method list: & 
      \mList & \in & \Classes \FnTo \PowersetFin\Methods\\
    method table: & 
      \mTable & \in & \Classes \times \Methods \FnToFin \Exprs\\
    \FJEUS program: & 
      \Prg & = & (\prec,\fList,\mList,\mTable)\\
  \end{array}
\end{displaymath}

\FJEUS is a language with nominal subtyping: $D\prec C$ means $D$ is
an immediate subclass of $C$.  The relation is well-formed if, when
restricted to $\Classes\setminus\{\NullType\}$ it is a tree successor
relation with root $\KObject$; thus, multiple inheritance is not
allowed. We write $\subclassOf$ for the reflexive and transitive
closure of $\prec$.  We also require $\NullType\subclassOf C$
for all $C\in \Classes$ and that $\KString$ is neither a subclass nor a
superclass of any proper class (other than $\NullType, \KObject, \KString$).

The functions $\fList$ and $\mList$ describe for each class~$C$ which
fields and method objects of that class have.  The functions are
well-formed if for all classes~$C$ and~$D$ such that $D\subclassOf C$,
$\fList(C)\subseteq \fList(D)$ and $\mList(C)\subseteq \mList(D)$,
i.e.\ classes inherit fields and methods from their superclasses. 
For $C\in\set{\NullType,\KString}$ we require that
$\fList(C)=\mList(C)=\emptyset$.
A method table $\mTable$ gives for each class and each method
identifier its implementation, i.e.\ the \FJEUS expression that forms
the method's body.

A method table is well-formed if the entry
$\mTable(C,m)$ is defined for all $m\in\mList(C)$.  
An implementation may be overridden in subclasses for the same number
of formal parameters. For simplicity, we do not include overloading.
We assume that the formal argument variables of a method $m$ are named
$x^m_1$, $x^m_2$, etc., besides the implicit and reserved variable
$\this$.  Only these variables may occur freely in the body of
$m$. The number of arguments of~$m$ is denoted~$\ar(m)$.

A class table $(F_0,M_0)$ models \FJEUS standard type
system, where types are simply classes.  The \emph{field typing} $F_0
: (\Classes \times \Fields) \FnToFin \Classes$ assigns to each class
$C$ and each field $f\in\mathit{fields}(C)$ the class of the field,
which is required to be invariant with respect to subclasses of $C$.
The \emph{method typing} $M_0 : (\Classes \times \Methods) \FnToFin
\Classes^*\times\Classes$ assigns to each class $C$ and each
method $m\in\mathit{methods}(C)$ a \emph{method type}, which specifies
the classes of the formal argument variables and of the result value.
%

We make explicit the notion of program point by annotating expressions
with \emph{expression labels} $i\in\Sites$: we write $[e]^i$ for
\FJEUS expressions, where $e$ is defined as before.  An \FJEUS program
is well-formed if each expression label $i$ appears at most once in
it.  In the following, we only consider well-formed programs, and
simply write $e$ instead $\LExpr{e}{i}$ if the expression label $i$ is
not important.


\subsection{Semantics}
A state consists of a store (variable environment or stack) and a heap
(memory).  Stores map variables to values, while heaps map locations
to objects.  The only kinds of values in \FJEUS are object locations
and $\NullVal$.  We distinguish two kinds of objects: \emph{ordinary
  objects} contain a class identifier and a valuation of the fields,
while \emph{string objects} are immutable character sequences tagged
with a word over the alphabet $\Sigma$.  The following table
summarizes this state model.
\begin{gather*}
  \begin{array}{>{$}r<{$}rclc>{$}r<{$}rcl}
    locations: & l &\in& \Locs  & \;\;\; &
    stores: & s & \in & \Vars \FnToFin \Values\\
    values: & v & \in & \Values = \Locs \uplus {\{\NullVal\}} &&
    heaps:  & h & \in & \Locs \FnToFin \Objs \uplus \SObjs\\
    string objects: & & & \SObjs = \Strings  \times \Sigma^*  &&
    objects:& & & \Objs = \Classes \times (\Fields \FnToFin \Values)\\  
  \end{array}
\end{gather*}

The \FJEUS semantics is defined as a big-step relation
$\FSem{s}{h}{e}{v}{h'}{w}$, which means that, in store $s$ and
heap~$h$, the expression $e$ evaluates to the value $v$ and modifies
the heap to $h'$, generating the event trace $w\in\Sigma^*$.  The
operational semantics rules can be found in Appendix~\ref{app:opsem}.

\section{Region-based Type and Effect Systems}

\subsection{Refined Types, Effects, and Type System Parameters}

\subsubsection{Refined Types}\label{sec:ref-types}

We assume a finite set~$\Regions$ of \emph{regions}, with $\Monoid\subseteq\Regions$.
We refine the standard object types by annotating objects with sets of regions:
\[
    \Types \ni \tau,\sigma \; ::= \; C_R
    \quad\quad\quad \text{where $C\in\Classes, R\subseteq\Regions$}
\]
such that if $C=\KString$ then $R\subseteq\Monoid$.

A value typed with $C_R$, with $R\cap\Monoid=\emptyset$,
intuitively means that it is a location pointing to an ordinary object of class
$C$ (or a subclass of $C$), and this location is abstracted to a
region $r\in R$, but no other region.
A value typed with $\KString_U$ (or $\KObject_U$ with
$U\subseteq\Monoid$) intuitively means that it is a location that
refers to a string object that is tagged with a word $w$ such that
$[w]\in U$.  We use subsets of~$\Regions$ rather than single elements
to account for joining branches of conditionals (including the
conditionals implicit in dynamic dispatch).

Since region sets are an over-approximation of the possible locations
where a non-string object resides, and string annotations are an
over-approximation of the tags, we define a subtyping
relation $\subtypeOf$ based on set inclusion:
\begin{mathpar}
  \Rule{}
       {C\subclassOf D \\ R\subseteq S}
       {C_R\subtypeOf D_S}
\end{mathpar}
The subtyping relation is extended to type sequences as expected:
$\seq\sigma \subtypeOf \seq\tau$ iff $|\seq\sigma|=|\seq\tau|$ and
$\sigma_i\subtypeOf \tau_i$, for all
$i\in\set{1,\ldots,|\seq\sigma|}$.

If $R$ is a singleton we call refined types $C_R$ 
\emph{atomic (refined) types} and denote the set of atomic types by
$\ATypes$. We often write~$C_r$ instead of $C_{\set{r}}$.
For a refined type $C_R$, let
$\atoms(C_R) := \setx{C_r}{r\in R}$. For a sequence
$\bar\tau=(\tau_1,\dots\tau_k)$ of refined types, let
$\atoms(\bar\tau) :=
\setx{(\sigma_1,\dots,\sigma_k)}{\sigma_i\in\atoms(\tau_i), \text{ for
    each $i$ with $1\leq i\leq k$}}$.

\subsubsection{Type and Effect Lattice}\label{sec:types-and-effects}
We lift the subtyping relation to include effects as well. 
We define the partial order $\sqsubseteq$ on
$\Lattice := \Types \times \Effects$
by $(C_{R}, U) \sqsubseteq (C'_{R'}, U')$ if and only if
$C_{R} \subtypeOf C'_{R'}$ and $U\subseteq U'$, for any
$C, C'\in \Classes$, $R, R'\subseteq\Regions$, and
$U, U' \subseteq \Monoid$.
Given two refined types $C_R$ and $D_S$, we define their \emph{join}
as $C_R\sqcup D_S=E_{R\cup S}$ where $E$ is the smallest (wrt
$\preceq$) common superclass\footnote{Note that such a class always
exists, as $C\preceq\KObject$, for any $C\in\Classes$.} of~$C$
and~$D$.
Given two elements $\ell,\ell'\in \Lattice$ with $\ell = (C_{R}, U)$ and
$\ell' = (C'_{R'}, U')$, we define their \emph{join}, denoted
$\ell\sqcup\ell'$, by $(C_R\sqcup C'_{R'},U\cup U')$.
Thus $(\Types,\sqcup)$ and $(\Lattice,\sqcup)$ are join-/upper-semilattices.
Given $T=\set{\tau_1,\dots,\tau_n}\subseteq\Types$ for some $n\geq 0$,
we denote by $\sqcup T$ the type
$\tau_1\sqcup \tau_2\sqcup\dots\sqcup\tau_n$, where by convention
$\sqcup T = \KObject_\Regions$ when $T=\emptyset$. For
$L\subseteq\Lattice$, the notation $\sqcup L$ is defined similarly.

\subsubsection{Parameters of the Type and Effect System}

The following table summarizes the parameters of our
system: a set of regions, a set of contexts, a context transfer
function, and an object abstraction function.
We explain them here briefly. 
\begin{displaymath}
  \begin{array}{>{$}r<{$}rcl}
    Regions (finite): & r,s,t& \in & \Regions\\
    Contexts (finite): & z & \in & \Contexts\\
    Context transfer function: & \phi & \in & \Contexts \times \Classes \times
    \Regions \times \Methods \times \Sites \FnTo \Contexts \\
    Object abstraction function: & \psi &  \in & \Contexts \times \Sites \FnTo \Regions\\
  \end{array}
\end{displaymath}

Regions are already defined in
Section~\ref{sec:types-and-effects}. They represent abstract memory
locations.  Each region stands for zero or more concrete locations.
Different regions represent disjoint sets of concrete locations, hence
they partition or color the memory.  Two pointers to different regions
can therefore never alias.
Thus the type system serves also as a unifying calculus for pointer
analysis.

Let us assume that we have the \emph{call graph} of a program.  A method
$m$ is then represented by a node in this graph, and a \emph{context}
corresponds to a possible path that leads to the node.  The
\emph{finite} set $\Contexts$ abstracts these paths. For example, it
can be chosen to consists of all locations in the program code or the
latter together with, say, the last 3 method calls on the stack
(3CFA). The meaning of these contexts is given by the functions $\phi$
and $\psi$ which we explain next.  The \emph{context transfer
  function} $\phi$ represents the edges in the abstract call graph. It
selects a context for the callee based on the caller's context, the
class of the receiver object, its region, the method name, and the
call site.
An \emph{object abstraction function} $\psi$ assigns a
region for a new object, given the allocation site and the current
method context. Notice that $\Contexts$ is an arbitrary finite set and
$\phi$, $\psi$ are arbitrary functions. Their choice does not affect
soundness but of course accuracy of the analysis. 
In \cite{Sma2011} a similar 
factorization of $\phi$ and $\psi$ for callee contexts and object
allocation names is presented.

\subsection{Declarative Type System}

As said earlier, the declarative type and effect system is general in the sense
that it produces method typings without considering contexts.
The method typings along with effects are computed 
with regard to only method signatures and the associated region information.
In addition, new objects are assigned arbitrary regions 
in a context-insensitive manner.   

The typing judgment takes the form $\FType{\Gamma}{e}{\tau}{U}$,
where~$e$ is an expression, the \emph{variable context} (\emph{store
  typing}) $\Gamma: \Vars \FnToFin \Types$ maps variables (at least
those in~$e$) to types, $\tau$~is a type, and~$U$ is
an element of $\Effects$.
The meaning is that, if the values of the variables comply
with~$\Gamma$ and the evaluation of $e$ terminates successfully, then
the result complies with~$\tau$, and the event trace generated during
this evaluation belongs to one of the equivalence classes in~$U$.
%
In particular, if $U\subseteq\Allowed$ then~$e$ adheres to the
guideline. It suffices to perform this latter check for an entry point
such as (the body of) a main method.

From a theoretical point of view, the declarative type system forms
the basis of our analysis. Once we prove its soundness
w.r.t.~operational semantics, the soundness of the semi-declarative
and algorithmic systems follows directly.

\subsubsection{Class Tables}\label{sec:ctables_decl}

A declarative class table $(\TPool_\mathsf{d},F,M_\mathsf{d})$ fixes a set of
$\TPool_\mathsf{d}\subseteq \ATypes$ of relevant atomic refined types and
models \FJEUS's class member types declaratively.  
The set $\TPool_\mathsf{d}$ is required to be closed under "supertyping", that
is, for any $C_r\in\TPool_\mathsf{d}$ and $D\in\Classes$ with $C\prec D$, we
have that $D_r\in\TPool_\mathsf{d}$.
One can often assume that the set~$\TPool_\mathsf{d}$ of relevant types
contains all types, i.e.~$\TPool_\mathsf{d} = \ATypes$. However, when we take
$\TPool_\mathsf{d}\subsetneq\ATypes$, by having $C\prec D$, $C_r\notin\TPool$,
and $D_r\in\TPool$, there is an implicit promise that an object with
type $D_r$ is never an object of type $C_r$ that has just been upcast.
The usefulness of this feature is illustrated in
Appendix~\ref{app:fjsec}.

The \emph{field typing} $F : (\Fields \times \ATypes) \FnToFin \Types$
assigns to each class $C$, region~$r$, and field $f\in\fList(C)$ the
type $F(f,C_r)$ of the field. The type is required to be invariant
with respect to subclasses of $C$. More formally, a field typing~$F$
is \emph{well-formed} if $F(f,D_r) = F(f,C_r)$, for all classes $C$,
subclasses $D\subclassOf C$, regions~$r$ with $D_r\in\TPool_\mathsf{d}$, and
fields $f\in\fList(C)$.  
For simplicity, in contrast to~\cite{comlan}, we do not use covariant
get-types and contravariant set-types for fields.

The \emph{declarative method typing} $M_\mathsf{d}: (\Methods \times \ATypes)
\FnToFin \Powerset{\Types^* \times \Types \times \Effects}$ assigns to
each class~$C$, region~$r$, and method~$m\in\mList(C)$, a set $M_\mathsf{d}(m,
C_r)$ of tuples $(\bar\sigma,\tau,U)$, where $\seq{\sigma}$ is a
sequence of atomic refined types for the methods' arguments, $\tau$ is
the refined type of the result value, and $U$ are the possible effects
of the method.
Every overriding method should be contravariant in the argument types,
covariant in the result class, and have a smaller effect
set. Formally, $M_\mathsf{d}$ is \emph{well-formed} if for all classes $C$,
subclasses $C'\preceq C$, regions $r$ with $C'_r\in\TPool_\mathsf{d}$, and
methods~$m\in\mList(C)$, it holds that
\[
\forall (\seq{\sigma}, \tau, U) \in M_\mathsf{d}(m, C_r). \; \exists (\seq{\sigma}', \tau', U')
\in M_\mathsf{d}(m, C'_r). \; (\seq{\sigma}', \tau', U') \msqsubseteq (\seq{\sigma}, \tau, U)  
\]
where we lift the partial order $\sqsubseteq$ to
$\Types^* \times \Types \times \Effects$ using
$(\bar\sigma', \tau',U') \msqsubseteq (\bar\sigma, \tau, U) $ iff
$\bar\sigma \subtypeOf \bar\sigma'$ and
$(\tau', U') \sqsubseteq (\tau, U)$.

Finally, all types occurring in the field and methods typings are
relevant. Formally, $C_r\in\TPool_\mathsf{d}$, for any~$C_r$ occurring in the
domain of $F$ or $M_\mathsf{d}$, and for any $C_r\in\atoms(C_R)$ with $C_R$
occurring in the image of $F$ or~$M_\mathsf{d}$.

\subsubsection{Type System}

For space reasons, we only present three of the type rules, given in
Figure~\ref{fig:type-system-declarative-selected}. The complete rules
are given in Appendix~\ref{app:decl}.
In the rule \rulename{TD-New}, we choose an arbitrary region as the
abstract location of the object allocated by this expression, as long
as the respective type is relevant.
For method calls, \rulename{TD-Invoke} requires that for all regions
$r\in R$ where the receiver object~$x$ may reside, there exists an 
entry in $M_\mathsf{d}$ for the called method and its class such that 
the resulting type and effect is suitable for the given argument
types and the expected result type and effect.
In the rule \rulename{TD-Builtin} we obtain the refined type of the
string returned by a call to the builtin method~$\fn$, by calling the
builtin method typing $M(\fn)$ on the tag abstractions of $\fn$'s
arguments.
Note that 
also denote by~$\sqsubseteq$ the
partial order over $\Powerset{\Monoid} \times \Effects$ defined by:
$(R', U') \sqsubseteq (R, U) $ iff $R' \subseteq R$ and
$U' \subseteq U$.

\begin{figure}[t]
\begin{mathpar}
  \Rule{TD-New}
  {C_{r} \in \TPool_\mathsf{d}}
  {\FType{\Gamma}{\New{C}}{C_{\set{r}}}{\set{\monneutral}}}
  \and
  \Rule{TD-Invoke}
  {\text{for all $r\in R$,
      there is $(\bar\sigma',\tau',U')\in M_\mathsf{d}(m,C_r)$}\\
      \text{such that $(\bar\sigma',\tau',U') \msqsubseteq \MethodTP{\seq\sigma}{\tau}{U}$}}
  {\FType{\Gamma,x:C_R,\seq y:\seq{\sigma}}{x.m(\seq y)}{\tau}{U}}
    \and
    \Rule{TD-Builtin}
         {\ar(\fn)=n \quad \text{$\Gamma(x_1) = \KString_{U_1}, \ldots, 
             \Gamma(x_n) = \KString_{U_n}$} \\\\  
           \text{$M(\fn) (u_1, \ldots, u_n) \sqsubseteq (U, U')$, for all $u_1 \in U_1, \ldots, u_n \in U_n$}}
         {\FType{\Gamma}{\fn(x_1, \dots, x_n)}{\KString_U}{U'}}
\end{mathpar}
\caption{Selected rules of the declarative type system.}
\label{fig:type-system-declarative-selected}
\end{figure}

An \FJEUS program $\Prg=(\prec, \fList, \mList, \mTable)$ is
\emph{well-typed} with respect to the class table $(\TPool_\mathsf{d},F,M_\mathsf{d})$
if for all classes~$C$, regions~$r$, methods~$m\in \mList(C)$, and
tuples $\MethodTP{\seq\sigma}{\tau}{U} \in M_\mathsf{d}(m, C_r)$, the judgment
$\FType{\Gamma}{\mTable(C,m)}{\tau}{U}$ can be derived with
$\Gamma = [\this\mapsto C_{r}] \cup
[x^m_i\mapsto\sigma_i]_{i\in\{1,\ldots,\ar(m)\}}$.

\subsubsection{Type System Soundness}
\label{sec:InterpretationAndSoundnessProof}

We state next the guarantees provided by the type system,
namely that if $\FType{\Gamma}{e}{\tau}{U}$ can be derived and
$U\subseteq\textsf{Allowed}$, then any event trace of the
expression~$e$ is allowed by the guideline.
See Appendix~\ref{app:proof} for a more general statement of
the soundness theorem and its proof.

\begin{theorem}[Soundness]\label{thm:decl-sound-shortened}
  Let $\Prg$ be a well-typed program and $e$ an expression with no
  free variables.
  If $\FType{\Gamma}{e}{\tau}{U}$ and 
  $\FSem{s}{h}{e}{v}{h'}{w}$, for some $\tau$, $U$, $v$, $h'$,
  and $w$, and with $\Gamma$, $s$, and $h$ the empty mappings, then
  $[w]\in U$.
\end{theorem}

\subsection{Semi-declarative Type System}
As already mentioned we rely on heuristic context information in order
to infer regions for newly created objects. That is to say, we use the
user-provided function $\psi$ in order to select a region for the
newly created object based on the position of the statement
(expression label) and the current context which is an abstraction of
the call stack.
Clearly, the use of such arbitrary user-provided decision functions
incurs an unavoidable loss of precision. The semi-declarative
type system which we now present precisely quantifies this loss of
precision. It is still declarative in the sense that types for methods
and classes can be arbitrary (sets of simple types), but it is algorithmic in that regions
for newly created objects are assigned using the function~$\psi$. It
also uses the equally user-provided function~$\phi$ to manage the
context abstractions.
The semi-declarative system is therefore sound
(Theorem~\ref{thm:sdecl-sound}) with respect to the declarative one
and thus also with respect to the operational semantics via
Theorem~\ref{thm:decl-sound-shortened}.

Further down, in Section~\ref{sec:algo} we will then give an
algorithmic type system which can be understood as a type inference
algorithm presented as a logic program. This algorithmic type system
will be shown sound (Theorem~\ref{thm:algo-sound}) \emph{and} complete
(Theorem~\ref{thm:algocomplete}) w.r.t.\ the semi-declarative system.

\subsubsection{Class Tables}\label{sec:ctables_semi}

We define next semi-declarative class tables $(\TPool_\mathsf{s},F,M_\mathsf{s})$.
%
The set of allowed refined types is parametrized by a context, that is,
$\TPool_\mathsf{s}$ is a function from~$\Contexts$ to $\ATypes$. Each set
$\TPool_\mathsf{s}(z)$ is closed under supertyping as in the declarative case.
The field typing~$F$ is as in the fully declarative case.
The \emph{semi-declarative method typing} $M_\mathsf{s}: (\Methods \times
\Contexts \times \ATypes) \FnToFin \Powerset{\Types^* \times \Types
  \times \Effects}$ assigns to each class~$C$, region~$r$, context~$z$,
and method~$m\in\mList(C)$ a set~$M_\mathsf{s}(m,z,C_r)$ of tuples
$(\bar\sigma,\tau,U)$ as in the fully declarative case.
As before, overriding methods have to satisfy the
following condition: for any context~$z$, atomic refined typed
$C'_r\in\TPool_\mathsf{s}(z)$, class~$C$ with $C'\preceq C$, and
method~$m\in\mList(C)$, it holds that
\[
\forall (\seq{\sigma}, \tau, U) \in M_\mathsf{s}(m, z, C_r).\; 
\exists (\seq{\sigma}', \tau', U')
\in M_\mathsf{s}(m, z, C'_r). \; (\seq{\sigma}', \tau', U') \msqsubseteq (\seq{\sigma}, \tau, U).
\]

\subsubsection{Typing Rules}

The parametric typing judgment takes the form
$\SType{\Gamma}{z}{e}{\tau}{U}$, where $\Gamma$, $e$, $\tau$, and $U$
are as for the declarative typing judgment
$\FType{\Gamma}{e}{\tau}{U}$, while $z\in\Contexts$ is a context.
The judgment has the same meaning as before, with the addition that it
is relative to the context~$z$.
The derivation rules are the same as for the declarative system (see
Appendix~\ref{app:decl}, Figure~\ref{fig:type-system-declarative-complete}), with the addition of the
context $z$ in each judgment, except for the two rules given next.
In the rule \rulename{TS-New}, we choose the region specified by~$\psi$
as the abstract location of the object allocated by this expression.
For method calls, \rulename{TS-Invoke} requires that for all regions
$r\in R$ where the receiver object~$x$ may reside, the method typing
in the context selected by~$\phi$ is suitable for the given argument
types and the expected result type and effect.
  \begin{mathpar}
    \Rule{TS-New}
         {r=\psi(z,i) \\ C_{r}\in\TPool_\mathsf{s}(z)}
         {\SType{\Gamma}{z}{\LExpr{\New C}{i}}{C_{r}}{\set{\monneutral}}}
    \and
    \Rule{TS-Invoke}
       {\text{for all $r\in R$,
           there is $(\bar\sigma',\tau',U')\in M_\mathsf{s}(m,z',C_r)$}
         \\
         \text{such that $(\bar\sigma',\tau',U') \msqsubseteq \MethodTP{\seq\sigma}{\tau}{U}$,
         where $z'=\phi(z,C,r,m,i)$}}
       {\SType{\Gamma,x:C_R,\seq y:\seq{\sigma}}{z}{\LExpr{x.m(\seq y)}{i}}{\tau}{U}}
  \end{mathpar}

A program $\Prg=(\prec, \fList, \mList, \mTable)$ is
\emph{well-typed} w.r.t.~the class table $(\TPool_\mathsf{s},F,M_\mathsf{s})$
if for all classes~$C$, contexts~$z$, regions~$r$,
methods~$m\in \mList(C)$, and tuples
$\MethodTP{\seq\sigma}{\tau}{U} \in M_\mathsf{s}(m,z,C_r)$, the judgment
$\SType{\Gamma}{z}{\mTable(C,m)}{\tau}{U}$ can be derived with
$\Gamma = [\this\mapsto C_{r}] \cup
[x^m_i\mapsto\sigma_i]_{i\in{\{1,\ldots,\ar(m)\}}}$.

Soundness of the semi-declarative type system follows directly from
the soundness of declarative type system. That is, since the rules of
semi-declarative system are obtained by adding context information to
the rules of the declarative system, the soundness result from the
previous subsection carries over here.

\begin{theorem}[Soundness]\label{thm:sdecl-sound}
  If an $\FJEUS$ program is well-typed with respect to a
  semi-declarative class table $(\TPool_\mathsf{s}, F, M_\mathsf{s})$, then it is
  well-typed with respect to the corresponding declarative class table
  $(\TPool_\mathsf{d}, F, M_\mathsf{d})$, where
  $\TPool_\mathsf{d} := \bigcup_{z\in\Contexts}\TPool_\mathsf{s}(z)$ and
  $M_\mathsf{d}(m,C_r) := \bigcup_{z\in\Contexts}M_\mathsf{s}(m,z,C_r)$.
\end{theorem}

Regarding the announced lack of completeness with respect to the declarative system, consider e.g.\ the case where $\psi$ returns one and the same region irrespective of context and position. In this case, two newly created objects:
\begin{lstlisting}
      StringBuffer x = new StringBuffer();
      StringBuffer y = new StringBuffer();
\end{lstlisting}
\noindent will be sent to the same region and, e.g., writing an unsanitized
string into~\ls{x} followed by outputting \ls{y} would be
overcautiously considered an error. More interesting examples would
involve a single allocation statement called several
times in different contexts. No matter how fine the abstraction
function is chosen, we can always find a situation where two different
allocations are sent to the same region because the two contexts are
identified by the context abstraction.

\subsection{Algorithmic Type System}\label{sec:algo}

The algorithmic type system is the one we use in 
our analysis. It returns the most precise typings and 
is complete with respect to the semi-declarative system.

\subsubsection{Class Tables}\label{sec:ctables_alg}

We defined next algorithmic class tables $(\TPool_\mathsf{a},F,M_\mathsf{a})$.  The sets
$\TPool_\mathsf{a}(z)$ of relevant refined types per context~$z$, and the field
typing $F$ are as in the semi-declarative case.
The \emph{algorithmic method typing}
$M_\mathsf{a}: (\Methods \times \Contexts \times \ATypes \times \ATypes^*)
\FnToFin \Types \times \Effects$ assigns to each class~$C$, region~$r$, context~$z$,
and method~$m\in\mList(C)$, and sequence~$\seq{\sigma}$ of atomic
refined types for the methods' arguments
(i.e.~$|\seq{\sigma}|=\ar(m)$), a type and effect
value~$M_\mathsf{a}(m,z,C_r,\seq{\sigma})$, which specifies the refined type of
the result value, as well as the possible effects of the method.
Also, as for the other type systems, we require that for any
context~$z$, atomic refined type $C'_r\in\TPool_\mathsf{a}(z)$, class~$C$ with
$C'\preceq C$, region~$r$, and method~$m\in\mList(C)$, and atomic type
sequences $\seq\sigma\subtypeOf\seq\sigma'$ with
$|\seq{\sigma}|=\ar(m)$, it holds that
$M_\mathsf{a}(m,z,C'_r,\seq{\sigma}') \sqsubseteq M_\mathsf{a}(m,z,C_r,\seq{\sigma})$.

\subsubsection{Typing Rules}\label{sec:rules_alg}

The algorithmic typing judgment $\AType{\Gamma}{z}{e}{\tau}{U}$ takes
the same form and has the same meaning as the semi-declarative
typing judgment.
The rules (see Appendix~\ref{app:alg-system},
Figure~\ref{fig:type-system-algorithmic-complete}) are in essence more
specialized versions of the ones in the semi-declarative system. For
instance, in the rule \rulename{TA-Invoke} we take the join of all
types which are computed with respect to all regions in which the
object $x$ may reside and all contexts returned by $\phi$.
\begin{mathpar}
  \Rule{TA-Invoke}
  {(\tau, U) = \bsqcup\setx{M_\mathsf{a}(m, z', C_r, \bar{\sigma})}{r\in R, z'=\phi(z,C,r,m,i)}}
  {\AType{\Gamma,x:C_R,\seq y:\seq{\sigma}}{z}{\LExpr{x.m(\seq y)}{i}}{\tau}{U}}
\end{mathpar}

An \FJEUS program $P=(\prec, \fList, \mList, \mTable)$ is
\emph{well-typed} with respect to the class table $(\TPool_\mathsf{a},F,M_\mathsf{a})$
if for all classes~$C$, contexts~$z$, regions~$r$,
methods~$m\in \mList(C)$, and sequence~$\seq{\sigma}$ of argument
types such that $(m,z, C_r, \seq{\sigma})\in\dom(M_\mathsf{a})$, the judgment
$\AType{\Gamma}{z}{\mTable(C,m)}{\tau}{U}$ can be derived with
$\Gamma = [\this\mapsto C_{r}] \cup
[x^m_i\mapsto\sigma_i]_{i\in{\{1,\ldots,\ar(m)\}}}$, where
$M_\mathsf{a}(m,z, C_r, \seq{\sigma})=(\tau, U)$.

\begin{theorem}[Soundness]\label{thm:algo-sound}
If an \FJEUS program $\Prg$ is well-typed with respect to algorithmic
class table $(\TPool_\mathsf{a}, F, M_\mathsf{a})$, then it is well-typed with respect
to the corresponding semi-declarative class table
$(\TPool_\mathsf{a}, F, M_\mathsf{s})$, where
$$M_\mathsf{s}(m,z,C_r) :=
\setx{(\seq{\sigma},\tau,U)}{(m,z,C_r,\bar{\sigma})\in\dom(M_\mathsf{a}),
  M_\mathsf{a}(m,z,C_r,\bar{\sigma})=(\tau,U)}.$$
\end{theorem}

Soundness of algorithmic type system is inherited from the soundness 
of semi-declarative system. 
%
Next, we state the completeness of the algorithmic system with respect
to the semi-declarative system 

\begin{theorem}[Completeness]\label{thm:algocomplete}
  Let $\Prg$ be a program and $(\TPool_\mathsf{s},F,M_\mathsf{s})$ a semi-declarative
  class table. If $\Prg$ is well-typed w.r.t. $(\TPool_\mathsf{s},F,M_\mathsf{s})$ then
  there is an algorithmic method typing $M_\mathsf{a}$ such that the following
  conditions hold:
\begin{itemize}
\item $\Prg$ is well-typed w.r.t. $(\TPool_\mathsf{s},F,M_\mathsf{a})$,
\item  ``$M_\mathsf{a}$ has better types than $M_\mathsf{s}$,'' that is, for
  each $(m,z,C_{r})\in\dom(M_\mathsf{s})$, each
  $(\bar\sigma,\tau,U)\in M_\mathsf{s}(m,z,C_{r})$, and each
  $\bar\sigma_\mathsf{a}\in\atoms(\bar\sigma)$, there is a
  $(\tau',U')\in M_\mathsf{a}(m,z,C_{r},\bar\sigma_\mathsf{a})$ such that
  $(\tau',U')\sqsubseteq (\tau,U)$.
\end{itemize}
\end{theorem}

The desired algorithmic typing $M_\mathsf{a}$ is constructed as the least fix-point 
of the operator that computes the most precise types of method bodies under
an assumed method typing. 
If we had not introduced $\psi$ to resolve the nondeterminism 
in the rule~\rulename{TD-New},
we would obtain a typing like $\Powerset{\Types\times \Effects}$
rather than $(\Types \times \Effects)$ for the algorithmic types, 
but then it would not be clear how to compare ``best'' typings
by iteration as we do. 
For one thing, the cardinality of  
$\Powerset{\Types\times \Effects}$ is exponentially
larger than that of $(\Types \times \Effects)$. More importantly
there is no obvious ordering on $\Powerset{\Types\times \Effects}$
to represent improvement. It seems that the automatic inference of regions 
without using contexts is a computationally harder problem with a disjunctive
flavor requiring for instance SAT-solving but not doable by plain fix-point 
iteration.

\subsubsection{Type Inference Algorithm}
\label{sec:alg}

From the algorithmic type system in Section~\ref{sec:algo}, a type inference
algorithm can easily be constructed by reading the rules as a functional program.
Appendix~\ref{app:alg} presents a more general type inference
algorithm that infers an algorithmic class table for a given
program~$P$, provided the standard Java types of the program's methods
and fields are also given in form of a class table~$(F_0,M_0)$.
As output it returns an algorithmic class table~$(\TPool_\mathsf{a},F,M_\mathsf{a})$
such that $P$ is well-typed with respect to it.
Thus the algorithm can be readily used to check whether an
expression~$e$ follows a guideline.


\section{Experimental Evaluation}\label{sec:eval}

\subsection{Implementation}\label{sec:impl}

We have implemented a variant of the type inference algorithm from
Section~\ref{sec:alg} which applies to actual Java programs rather
than \FJEUS.\footnote{The implementation is available at \url{https://github.com/ezal/TSA}.} We describe next the main
ingredients of the implementation (see Appendix~\ref{app:eval} for further details).
Most importantly, we phrase our analysis as a dataflow problem; it is
well-known that type inference can be formulated as a dataflow
problem, see e.g.~\cite{nielsonbook}.
We thus distinguish between intraprocedural analysis and
interprocedural analysis.

Our implementation is built on top of the Soot
framework~\cite{sootVallee,lam2011soot}.
We benefit from Soot in two ways. First, we use it to transform Java
code into Jimple code (a model for Java bytecode), which
represents the source language of our analysis.
Second, we use Soot's generic intraprocedural dataflow analysis
(extending the \ls{ForwardFlowAnalysis} class) to implement
our intraprocedural analysis.
The interprocedural analysis is implemented through a standard
fix-point iteration using summary tables. Namely, each iteration
starts with a summary table, returns a new table, and the two tables
are compared; if there is a difference the old one is replaced with
the new one and the next iteration starts with the latter. It is clear
that the new table extends the old one at each iteration
so that a fix-point is reached and hence the analysis
terminates~\cite{SharirPnueli}.

\subsection{Experiments}

We tested and evaluated our tool on the Stanford Securibench Micro
benchmark.\footnote{\url{https://suif.stanford.edu/~livshits/work/securibench-micro/}}
Among the 12 categories of test cases provided by the benchmark, we
have analyzed all of them, excluding the one on reflection.
Table~\ref{tab:securibench-micro} lists the results obtained.
The table also contains a row for the additional examples we
considered, which include the ones appearing in this paper.
The 't' and 'w' columns denote respectively the number of tests in the
category, and how many of those run as expected.
Whenever the result is not as expected, the reason is mentioned in the
``comments'' column, as follows: 
\begin{compactenum}[(1)]
\item\label{item:path_sensitivity} Detection of the problem in the
  test case requires path-sensitivity, while our analysis is
  path-insensitive.
\item\label{item:strong_updates} Field updates are conservatively
  treated as weak updates, which sometimes leads to false positives,
  see Appendix~\ref{app:strongupdates3} for a concrete example.
\item\label{item:aliasing3} We believe that this test case was wrongly
  marked as violating by the benchmark, see
  Appendix~\ref{app:aliasing3} for details.
\item\label{item:features} We do not yet support
  two-dimensional arrays and concurrent features.
\end{compactenum}
Each test case is analyzed in at most a few seconds (on a
standard computer), except for one case which required 18 seconds.
\begin{table}[t]
\centering
\caption{Results on the SecuriBench Micro benchmark and on additional examples.}
\label{tab:securibench-micro}
\begin{tabular}{|l|r@{/}l|l|@{\qquad}|l|r@{/}l|l|}
  \hline
  test category & w&t & comments
  &
  test category & w&t & comments
  \\\hline\hline
  Aliasing & 5&6 & (\ref{item:aliasing3}) 
  & 
  Inter & 14&14 &
  \\\hline
  Arrays & 8&10 & (\ref{item:features}): matrices
  & 
  Pred & 7&9 & (\ref{item:path_sensitivity})
  \\\hline
  Basic & 42&42 & 
  & 
  Sanitizers & 6&6 &
  \\\hline
  Collections & 14&14 & 
  & 
  Sessions & 3&3 &
  \\\hline
  DataStructures & 6&6 &
  & 
  StrongUpdates & 3&5 & (\ref{item:strong_updates}), (\ref{item:features}): \texttt{synchronize}
  \\\hline
  Factories & 3&3 &
  &
  our examples & 24&25 & (\ref{item:path_sensitivity}) 
  \\\hline
\end{tabular}
\end{table}
We have also successfully analyzed (in 0.3sec) an application from the
Stanford Securibench
benchmark,\footnote{\url{https://suif.stanford.edu/~livshits/work/securibench/}}
namely \strcl{blueblog}.
%

\section{Related Work} 

We present a review of recent work in the literature that is
relevant to our work. 
Static analysis has a long history as a research area, which has 
also been subject to interest from industry. Among many books and 
surveys available in the literature, we refer to \cite{nielsonbook} 
for fundamentals and to~\cite{Chess2007} for an application-oriented 
reference. In~\cite{Chess2007} authors give
a detailed explanation of static analysis as part of code review,
explain secure programming guidelines and provide exercises with 
the Fortify code analyzer~\cite{fortifyurl}.
%
Other commercial static analysis tools that help with secure web
programming in particular include 
CheckMarx~\cite{checkmarxurl},
%
AppScan~\cite{appscanurl},
%
and Coverity~\cite{coverityurl}.
These tools check source code against vulnerabilities for various languages 
including C and Java, and provide compliance with guidelines offered by 
institutions such as OWASP, MISRA, SANS, and Mitre CWE at different
levels. Although this fact is expressed in their data sheets, the
question of how the commercial tools formalize the guidelines differs
from one tool to another and the common practice in general is to
hardwire a given guideline into the tool.
One of our partially reached goals is to formalize guidelines so that
they are specified separately from the source code of our tool and are
open to independent review.

Skalka et al.\ develop a type and effect inference system to enforce
trace-based safety properties for
higher-order~\cite{DBLP:journals/jfp/SkalkaSH08} and object-oriented
\cite{DBLP:journals/entcs/SkalkaSH05,DBLP:journals/lisp/Skalka08}
languages. They represent event traces of programs as effects and
compute a set of constraints which are later fed into off-the-shelf
model-checking tools. Thus combining type inference and model
checking, Skalka et al. are able to analyze programs with respect to
trace properties, in particular flow-sensitive security properties
such as access control and resource usage policies.
The difference to our approach is that types do not contain the full
information about the possible traces but only a finite abstraction
which is just fine enough to decide compliance with a given policy. In
this way, one may expect more succinct types, more efficient
inference, and better interaction with the user. Another difference is
that our approach is entirely based on type systems and abstract
interpretation and as such does not rely on external model-checking
software. This might make it easier to integrate our approach with
certification. One can also argue that our approach is more in line
with classical type and effect systems where types do not contain
programs either but rather succinct abstractions akin to our effects.
Our work performs string analysis by reducing the problem to a type
inference analysis, wherein to track strings types are extended with
regions.

Closer to the present work is the Java(X) system
\cite{Degen11,DBLP:conf/ecoop/DegenTW07}. It uses ML-style type
inference based on parametrically polymorphic (refined!) types
rather than full polymorphism as we do (polymorphic in the sense that a method
type is an arbitrary set of simple, refined types). Also the system
does not have region-based tracking of aliasing. On the other hand,
Java(X) supports \emph{type state} \cite{DBLP:conf/ecoop/DeLineF04},
i.e.\ the ability of changing the type of an object or a variable
through a modifying action. For this to be sound it is clear that the
resource in question must be referenced by a unique pointer and a
linear typing discipline is used in loc.cit.\ to ensure this. We think
that type state is essentially orthogonal to our approach and could be
added if needed. So far, we felt that for the purpose of enforcing
guidelines for secure web programming type state is not so
relevant. 

\section{Conclusion}

We have developed a type-based analysis to check Java%
\footnote{As usual, the formal description of our analysis is in terms
  of an idealized language, \FJEUS.
  The implementation takes genuine Java programs. However, it does not
  support certain features such as concurrency and reflection, which
  we discuss below.} %
programs against secure programming guidelines.
The analysis is sound, albeit incomplete: if the analysis reports
``success'', then this means the given program follows the guideline,
if it reports ``failure'' the program may or may not follow the
guideline.

Our system is based on the region-based type system in~\cite{comlan}
and extends it in several directions.
First, we enhanced the refined type system in~\cite{comlan} with an
effect mechanism to track event traces, following~\cite{fast}.
Second, we parametrized the system over arbitrary guidelines using a)
syntactic monoids to abstract traces and string values, b) external
static methods given by their semantics, and c) using mockup code to
represent library components.
Third, we provide a more precise polymorphism for method typings via a
precise interprocedural analysis (implemented based on our algorithmic
type and effect system) while the system in~\cite{comlan} mainly uses
contexts to index refined types and gives only a limited amount of
polymorphism requiring user intervention. In~\cite{comlan} contexts
provided a novel rigorous type-theoretic foundation for
context-sensitive and points-to analyses.  We still use contexts to
determine regions for newly allocated objects.
Fourth, we provide a clear semantic foundation for our analysis and
obtain the implementation of type inference in three steps, namely via
declarative, semi-declarative, and algorithmic type systems.  We
establish correctness of our system by proving that each type system
is sound with respect to its precursor and that the algorithmic system
is also complete w.r.t.~the semi-declarative system.
Finally, we implemented our system allowing us to analyze actual Java
code. We rely on the Soot Framework and analyze on the level of Soot
intermediate code (i.e.~Jimple) which gives us a variety of language
features such as various loop and branching constructs, as well as a
limited amount of exception handling for free.
While our implementation is for now a prototype, it has allowed us to
check a significant part of the SecuriBench Micro benchmark and also a
medium-size application.


It is possible to extend our work in several ways. The obvious
direction is to formalize more guidelines taken, e.g., from the OWASP
or SANS portals to secure web programming. This will further validate
our approach to formalization through the described mix of automata,
monoids, and modelling of framework code.  It will also motivate
various extensions to our type-based analysis which we will tackle as
needed. We describe here the most important ones.

\newcommand{\myparagraph}[1]{\smallskip\noindent\textit{#1}.\ }

\myparagraph{Reflection}
While unconstrained use of reflection would seem to
preclude any kind of meaningful static analysis it appears that the
use of reflection in actual applications is rather restricted. Being
able, for example, to integrate into our analysis the contents of
XML-manifests and similar data which is usually processed via
reflection would carry a long way.

\myparagraph{Path sensitivity}
If needed, a limited amount of path sensitivity can be obtained by
introducing a class of booleans with subclasses representing true and
false. One could then use refined typings for certain predicates which
depending on the \emph{refined} types of their parameters ensure that
the result is true or false.

\myparagraph{More general automata}
Some guidelines may require us to look beyond regular properties.
This may require having to use richer specification formalisms, like
automata over infinite alphabets~\cite{KaminskiF94}, to capture
authorization per resource rather than per resource kind.
Also, we might want to analyze infinitary program properties,
e.g. liveness and fairness properties of infinite program traces. To
this end it should be possible to replace our finite state machines
with B\"uchi automata, and accordingly, our monoids with so-called
$\omega$-semigroups. We have successfully carried out initial
investigations along those lines~\cite{HofmannChen2014,HofmannLedent2017}.

\myparagraph{Concurrency and higher-order} 
Other possibilities for extension are concurrency and higher-order
functions, i.e.~anonymous methods, inner classes, and similar. We
believe that ideas from higher-order model checking in type-theoretic
form~\cite{DBLP:conf/lics/KobayashiO09} could be fruitful there and
initial investigations 
have confirmed this.

\myparagraph{IFDS} Another interesting extension would be to harness the
approach to interprocedural analysis IFDS
\cite{RepsHS95} for our type inference which has been
integrated with Soot in the form of the Heros plugin. While this would
not extend expressivity or accuracy of our approach the increased
efficiency might increase its range. As we see it, the gist of IFDS is
on the one hand the restriction of function summaries to atoms using
distributivity and on the other a clever management of updates to
summaries so that repeated recomputation during fixpoint iteration can
be reduced to a minimum. Our approach already incorporates the first
aspect because types of callees and method parameters are always
atomic. We think that this is the reason why our approach works
surprisingly well. Nevertheless, we hope that some more efficiency
could be squeezed out of the second ingredient and it would be
particularly interesting to use Heros as a kind of 
solver to which we can offload the computation of the final method
table.

\myparagraph{Certification} While we have not actually fleshed this out, it is clear that the type-theoretic formulation of our analysis lends itself particularly well to independent certification.  Indeed, we could write a type checker (which would involve no fixpoint iteration and similar algorithmic techniques at all) for the declarative system and then have the algorithmic type inference generate a derivation in the former system to be checked. This might be much easier than verifying the implementation of the inference engine. Similarly, our soundness proof can be formalized in a theorem prover and declarative typing derivations can then be used to generate formal proofs of correctness; see e.g.\ 
the Mobius project \cite{BartheBCGHMPPSV06}.


\bibliographystyle{abbrv}
\bibliography{xss,own,region}

\newpage
\appendix

\section{Guideline for Authorization}\label{app:auth}

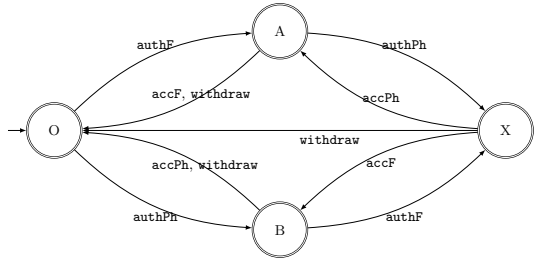
\begin{figure}[t]
  \centering
  \scalebox{0.6}{
    \begin{tikzpicture}[>=latex,initial text={},
                      every state/.style={minimum width=1.2cm}]
    \node[state,initial,accepting] (o)   at (1.0,2.5) {O};
    \node[state,accepting]         (a)   at (6,4.7) {A};
    \node[state,accepting]         (b)   at (6,0.3) {B};
    \node[state,accepting]         (x)   at (11,2.5) {X};
    \path[->] (o)   edge [bend left=20, above] 
                       node {\strcl{authF}} (a);
    \path[->] (a)   edge [bend left=20, below, near start] 
                       node {\parbox{3.0cm}{\strcl{accF}, \strcl{withdraw}}} (o);
    \path[->] (o)   edge [bend right=20, below] 
                       node {\strcl{authPh}} (b);
    \path[->] (b)   edge [bend right=20, above, near start] 
                       node {\parbox{3.0cm}{\strcl{accPh}, \strcl{withdraw}}} (o);
    \path[->] (a)   edge [bend left=20, above] 
                       node {\strcl{authPh}} (x);
    \path[->] (x)   edge [bend left=20, above] 
                       node {\strcl{accPh}} (a);
    \path[->] (b)   edge [bend right=20, below] 
                       node {\strcl{authF}} (x);
    \path[->] (x)   edge [bend right=20, below] 
                       node {\strcl{accF}} (b);
    \path[->] (x)   edge [below, near start] 
                    node {\parbox{3.5cm}{\strcl{withdraw}}} (o); 
  \end{tikzpicture}}
  \caption{Authorization policy automaton}
  \label{fig:automaton-auth}
\end{figure}

Let us look at another guideline which states that any access to sensitive data must only be done after authorization. In this case, the programmer should call authorization methods before access methods.

The following code fragment illustrates the distinction between accesses to different types of resources, namely regular files and phone directories. The guideline is not satisfied.

\begin{minipage}{0.55\textwidth}
  \begin{lstlisting}[language=java]
class Authorization {
   main () {
      File file = new File("file.txt");
      Phone phone = new Phone("phonebook.txt");
      authFile();
      file.access();       /* OK */
      withdrawAuth();
      file.access();       /* BAD */
      authPhone();
      phone.access();      /* OK */
   }
}
\end{lstlisting}
\end{minipage}

To model this guideline we first consider the possible states which the
execution of this program can reach and then the events that trigger
state transitions.  We define four states; in state \strcl{O} no
access is allowed, in state \strcl{A} file access is allowed, in state
\strcl{B} phone access is allowed and in \strcl{X} both kinds of
access are allowed.  We assume that successful calls to framework
methods given above issue the following events: \ls{authF},
\ls{accF}, \ls{authPh}, \ls{accPh}, and \ls{withdraw}. Now we
can specify the policy with the finite state machine in
Figure~\ref{fig:automaton-auth}. 

Notice that this example is a rather coarse oversimplification. A more
realistic version could have different kinds of authorization that may
also be withdrawn via system calls. Also, we are not aiming at
implementing an access control system, but rather to check that code
written with good intentions does not violate guidelines pertaining to
authorization.

\section{A More General Guideline for Taintedness}\label{app:taint}

We reproduce here a guideline first formalized in~\cite{fast}. 

In addition to \ls{getString} and \ls{putString}
from Example~\ref{ex:intro}, we also have a function
\ls{escapeToHtml(input)} which escapes '\ls{<}' and '\ls{>}' to
'\ls{&lt;}' and '\ls{&gt;}', and a function
\ls{escapeToJs(input)} escaping user input to be used in a JavaScript
context which we don't detail. The tag alphabet is $\Sigma=\{\strcl{Lit},
\strcl{C1}, \strcl{C2}, \strcl{Script}, \strcl{/Script},
\strcl{Input}\}$ with the intention that $\strcl{Input}$ tags
unsanitized user input, \ls{escapeToHtml} results are tagged
\ls{C1}, \ls{escapeToJs} results are tagged
\ls{C2}, the tokens \ls{"<script>"} and \ls{"</script>"} are tagged
eponymously, and everything else is tagged~\ls{Lit}.

The automaton from Figure~\ref{fig:automaton} then specifies the guideline.
\begin{figure}[t]
  \centering
  \scalebox{0.8}{
  \begin{tikzpicture}[>=latex,initial text={},
                      every state/.style={minimum width=1.5cm}]
    \node[state,initial,accepting] (html)   at (1.5,1.5) {HTML};
    \node[state,accepting]         (script) at (8.5,1.5) {SCRIPT};
    \node[state]                   (fail)   at (5,0.3) {FAIL};
    \path[->] (html)   edge [loop below,left,near end] 
       node {\parbox{1.2cm}{\strcl{C1},\strcl{Lit},\\\strcl{/Script}}} 
       (html);
    \path[->] (html)   edge [above] 
                       node {\strcl{Script}} (script);
    \path[->] (script) edge [bend right=15,above] 
                       node {\strcl{/Script}} (html);
    \path[->] (script) edge [loop below,right,near start]
       node {\parbox{1.2cm}{\strcl{C2},\strcl{Lit},\\\strcl{Script}}}
       (script);
    \path[->] (html)   edge [above]
                       node {\strcl{C2}} (fail);
    \path[->] (html)   edge [below,bend right=30]
                       node {\strcl{Input}} (fail);
    \path[->] (script) edge [above]
                       node {\strcl{C1}} (fail);
    \path[->] (script) edge [below,bend left=30]
                       node {\strcl{Input}} (fail);
    \path[->] (fail)   edge [loop below,right,near start] 
       node {\parbox{3.5cm}{\strcl{C1},\strcl{C2},\strcl{Lit},\\%
                            \strcl{Script},\strcl{/Script},\strcl{Input}}}
       (fail);
  \end{tikzpicture}}
  \caption{Policy automaton for taintedness under sanitization}
  \label{fig:automaton}
\end{figure}
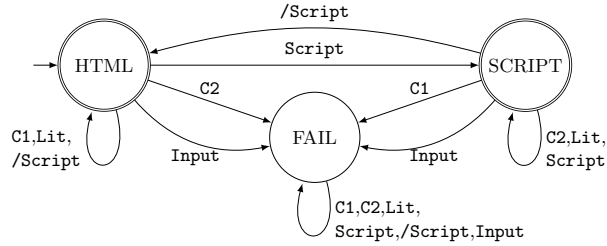
Intuitively, the automaton is run on the sequence of tags representing
how the string fed to \ls{putString} was obtained and should never go
into the failure state.

\section{Additional Details on \FJEUS Syntax}\label{app:syntax}


Consider the program in Example~\ref{ex:intro}. 
The program in \FJEUS and its class table is shown below.

\begin{minipage}{0.5\textwidth}
\begin{displaymath}  
  \begin{array}{rcl}
\mathit{fields}(\ident{C}) & = & \emptyset\\
\mathit{fields}(\ident{D}) & = & \{\ident{s}\}\\
\mathit{methods}(\ident{C}) & = & \{\ident{main}\}\\
\mathit{methods}(\ident{D}) & = & \{\ident{cD}\}\\
\mathit{mtable}(\ident{C},\ident{main}) & = & e_{\ident{main}}\\
\mathit{mtable}(\ident{D},\ident{cD}) & = & \ident{this}.s := x^{\ident{cD}}_1
\\[0.5ex]
\mathit{F_0}(\ident{D},\ident{s}) & = & \KString\\
\mathit{M_0}(\ident{C},\ident{main}) & = & (\varepsilon,\KString)\\
\mathit{M_0}(\ident{D},\ident{cD}) & = & (\seqvar{\KString},\KString)\\
  \end{array}
\end{displaymath}
\end{minipage}
\begin{minipage}{0.5\textwidth}
  \begin{lstlisting}[language=fjeus]
$e_{\ident{main}}$ :=
  let $a$ = $\ident{getString}$() in
  let $b$ = "test" in
  let $f_1$ = new $D$ in
  let _ = $f_1.\ident{cD}(a)$ in
  let $f_2$ = new $D$ in
  let _ = $f_2.\ident{cD}(b)$ in
  $\ident{putString}(f_2.s)$
$\phantom{I}$
\end{lstlisting}
\end{minipage}

\noindent
In the $e_{\ident{main}}$ expression above the underscore denotes some
variable not used in the rest of the program.
Note that, by convention, methods that normally would not return a
value, are assumed to return an empty string.

\section{Additional Details on \FJEUS Semantics}\label{app:opsem}

\begin{figure}
\centering
\begin{mathpar}
  \Rule{}
       { }
       {\FSem{s}{h}{x}{s(x)}{h}{\varepsilon}}
  \and 
  \Rule{}
       { }
       {\FSem{s}{h}{\Null}{\NullVal}{h}{\varepsilon}}
  \and
  \Rule{}
       {s(x)=s(y) \\ \FSem{s}{h}{e_1}{v}{h'}{w}}
       {\FSem{s}{h}{\IfEqual{x}{y}{e_1}{e_2}}{v}{h'}{w}}
  \and
  \Rule{}
       {s(x) \ne s(y) \\ \FSem{s}{h}{e_2}{v}{h'}{w}}
       {\FSem{s}{h}{\IfEqual{x}{y}{e_1}{e_2}}{v}{h'}{w}}
  \and
  \Rule{}
       {\FSem{s}{h}{e_1}{v_1}{h_1}{w_1} \\
        \FSem{s[x\mapsto v_1]}{h_1}{e_2}{v_2}{h_2}{w_2}} 
       {\FSem{s}{h}{\Let{x}{e_1}{e_2}}{v_2}{h_2}{w_1\cdot w_2}}
  \and
  \Rule{}
       {l\not\in\domof(h) \\
        F = [f\mapsto \NullVal]_{f\in \fList(C)}} 
       {\FSem{s}{h}{\New C}{l}{h[l\mapsto (C,F)]}{\varepsilon}}
  \and
  \Rule{}
       {\FSem{s}{h}{e}{v}{h}{w} \\
       \classOfVal_h(v)\preceq C}
       {\FSem{s}{h}{\Cast{e}{C}}{v}{h}{w}}
  \and
  \Rule{}
       {s(x) = l \\ h(l) = (\_,F)}
       {\FSem{s}{h}{x.f}{F(f)}{h}{\varepsilon}}
  \and
  \Rule{}
       {s(x)=l \\ h(l)=(C,F) \\\\ h'=h[l\mapsto(C,F[f\mapsto s(y)])]} 
       {\FSem{s}{h}{\Assign {x.f}{y}}{s(y)}{h'}{\varepsilon}}
  \and
  \Rule{}
       {s(x)=l \\ s' = [\this\mapsto l]\cup [x^m_i \mapsto
        s(y_i)]_{i\in\{1,\ldots,n\}} \\\\
        h(l)=(C,\_) \\ \FSem{s'}{h}{\mTable(C,m)}{v}{h'}{w}}
      {\FSem{s}{h}{x.m(y_1,\dots,y_n)}{v}{h'}{w}}
  \and
  \Rule{}
       {h(s(y_i)) = (\Str_i, w_i) \\ 
        l\not\in\domof(h) \\ h' = h[l\mapsto (\Str,w)] \\\\
        \sem(\fn)((\Str_1,w_1) \ldots, (\Str_n,w_n)) \ni (\Str,w,w')}
       {{\FSem{s}{h}{\fn(y_1, \ldots, y_n)}{l}{h'}{w'}}}
  \and
  \Rule{}
       {l\not\in\domof(h) \\\\
        h' = h[l\mapsto (\Str, \wordoflit(\Str))]}
       {\FSem{s}{h}{"\Str"}{l}{h'}{\varepsilon}}
  \and
  \Rule{}
       {h(s(x_1)) = (\Str_1, w_1) \\
        h(s(x_2)) = (\Str_2, w_2) \\\\
        l\not\in\domof(h) \\
        h' = h[l\mapsto (\Str_1\cdot\Str_2, w_1 \cdot w_2)]}
       {\FSem{s}{h}{x_1+x_2}{l}{h'}{\varepsilon}}
\end{mathpar}
\caption{The \FJEUS operational semantics}
\label{fig:fjeus-opsem-complete}
\end{figure}

The complete rules for the operational semantics can be found in
Figure~\ref{fig:fjeus-opsem-complete}.
A premise involving a partial function, like $s(x)=l$, 
always implies the side condition $x\in\domof(s)$.

The rules for variables, $\Null$, and the conditional are standard.
For \syntax{let} expressions, we concatenate the event traces of the
two expressions.
A new object is allocated at a fresh location with all fields set to
$\NullVal$.
A cast just checks whether the target type is a subtype of the subject
type and raises an exception if that is not the case. The auxiliary
function $\classOfVal_h(v)$ determines the type of the subject value
$v$ as follows: $\classOfVal_h(v)=\NullType$ if $v=\NullVal$, it
is $C$ if $h(v)=(C,\_)\in\Objs$, and it is $\KString$ if $h(v)\in
\SObjs$.
A field read access returns the field contents, while a field write
access updates the heap accordingly (and also evaluates to the written
value).

At a method call, a new store is created, consisting of a special
variable $\this$ and of the method parameters initialized with the
values of the passed arguments.  The return value, final heap and
event trace of the method execution are also the result of the call.
A builtin method call allocates a fresh location which stores the
result obtained by invoking directly the external semantics of the
builtin method on its string arguments.

In the rule for string literals, we rely on $\wordoflit$ to tag the
new string object.  The tagging of strings is a homomorphism with
respect to string concatenation.  

In non-recursive rules, except the one for builtin method calls, we
have the empty trace $\varepsilon$ as the generated event trace, as the
corresponding expressions do not produce any output.


\section{Additional Details on the Type Systems}\label{app:types}

\subsection{Declarative Type and Effect System}\label{app:decl}

The complete rules for the \FJEUS parametric type system can be found
in Figure~\ref{fig:type-system-declarative-complete}.  
We briefly explain the rules not presented in the body of the paper.
The rule \rulename{TD-Sub} is used to obtain weaker types and effects
for an expression.
The rule \rulename{TD-Var} looks up the type of a variable
in the context $\Gamma$. 
In \rulename{TD-If}, we require that both branches have the same
type and effect. This can be obtained in conjunction with the
weakening rule \rulename{TD-Sub}.
Furthermore, the rule exploits the fact that the two variables must
point to the same object (or be null) in the \syntax{then} branch,
therefore the intersection of the region sets can be assumed.
In \rulename{TD-Let}, we take into account that first the effects of
expression $e_1$ take place, and then the effects of expression
$e_2$. The overall effect is thus the concatenation of the sub-effects,
where we define $UU' = \setx{u\cdot u'}{u\in U, u'\in U'}$.

When a $C$ object is null, denoted as $\Null$, it is given type
$\NullType$ with an empty region and empty effect via the rule
\rulename{TD-Null}. This is done so that $\Null$ can be cast to any
object (in particular~$C$) in whatever region, by also using the
\rulename{TD-Sub} rule. Note that our type system is not meant to
prevent null-pointer exceptions, however, being an extension of the
standard Java type system, it prevents ``method not understood'' and
similar runtime errors. All soundness statements are conditional on
standard termination, i.e.\ termination without runtime
errors.\footnote{We do not include exceptions in \FJEUS and thus treat
  any exception as a runtime error. However, the Soot frontend that we
  use in the implementation (see Section~\ref{sec:impl}) compiles away
  a certain amount of exceptions.}
The rule \rulename{TD-Cast} allows casting the type of an expression
to a subtype or a supertype.

When a field is read (\rulename{TD-GetF}), we look up the type of the
field in the $F$~table. As the variable $x$ may point to a number of
regions, we need to ensure that $\tau$ is an upper bound of the types
of $f$ over all $r\in R$.  In contrast, when a field is written
(\rulename{TD-SetF}), the written value must have a subtype of the
types allowed for that field by the $F$~table with respect to each
possible region $r\in R$.

In the rule \rulename{TD-Lit}, we use $\wordoflit$ to determine the
annotation for string literals.
In \rulename{TD-Concat}, the type of a string obtained by
concatenation is defined by concatenating the monoid elements of the
types of the concatenated strings.

\begin{figure}
  \begin{mathpar}
    \Rule{TD-Sub}
         {\FType{\Gamma}{e}{\tau}{U} \\ \tau \subtypeOf \tau' \\ 
          U\subseteq U'}
         {\FType{\Gamma}{e}{\tau'}{U'}}
    \and
    \Rule{TD-Var}
         { }
         {\FType{\Gamma,x:\tau}{x}{\tau}{\{\monneutral\}}}
    \and
    \Rule{TD-If}
         {\FType{\Gamma, x:C_{R\cap S}, y:D_{R\cap S}}{e_1}{\tau}{U} \\ 
          \FType{\Gamma, x:C_{R}, y:D_{S}}{e_2}{\tau}{U} \\
        }
         {\FType{\Gamma, x:C_{R}, y:D_{S}}{\IfEqual{x}{y}{e_1}{e_2}}{\tau}{U}}
    \and
    \Rule{TD-Let}
         {\FType{\Gamma}{e_1}{\tau}{U} \\
           \FType{\Gamma,x:\tau}{e_2}{\tau'}{U'}}
         {\FType{\Gamma}{\Let{x}{e_1}{e_2}}{\tau'}{UU'}}
    \and
    \Rule{TD-Null}
         { }
         {\FType{\Gamma}{\Null}{\NullType_\emptyset}{\{\monneutral\}}}
    \and
    \Rule{TD-New}
         {C_{r} \in \TPool_\mathsf{d}}
         {\FType{\Gamma}{\New{C}}{C_{\set{r}}}{\set{\monneutral}}}
    \and
    \Rule{TD-Cast}
         {\FType{\Gamma}{e}{C_R}{U} \\ C\preceq D \text{ or } D\preceq C}
         {\FType{\Gamma}{\Cast{e}{D}}{D_R}{U}}
    \and
    \Rule{TD-Invoke}
         {\text{for all $r\in R$,
             there is $(\bar\sigma',\tau',U')\in M_\mathsf{d}(m,C_r)$
             such that $(\bar\sigma',\tau',U') \msqsubseteq \MethodTP{\seq\sigma}{\tau}{U}$}}
       {\FType{\Gamma,x:C_R,\seq y:\seq{\sigma}}{x.m(\seq y)}{\tau}{U}}
    \and
    \Rule{TD-Builtin}
         {\ar(\fn)=n \quad \text{$\Gamma(x_1) = \KString_{U_1}, \ldots, 
             \Gamma(x_n) = \KString_{U_n}$} \\\\  
           \text{$M(\fn) (u_1, \ldots, u_n) \sqsubseteq (U, U')$, for all $u_1 \in U_1, \ldots, u_n \in U_n$}}
         {\FType{\Gamma}{\fn(x_1, \dots, x_n)}{\KString_U}{U'}}
    \and
    \Rule{TD-GetF}
         {
           F(f,C_r)\subtypeOf\tau,\ \text{for all $r\in R$}}
         {\FType{\Gamma,x:C_R}{x.f}{\tau}{\{\monneutral\}}}
    \and
    \Rule{TD-SetF}
         {
           \tau \subtypeOf F(f,C_r),\ \text{for all $r\in R$}}
         {\FType{\Gamma,x:C,y:\tau}{\Assign{x.f}{y}}{\tau}{\{\monneutral\}}}
    \and
    \Rule{TD-Lit}
         {\wordoflit(\Str) = w}
         {\FType{\Gamma}{"\Str"}{\KString_{\{[w]\}}}{\{\monneutral\}}}
    \and
    \Rule{TD-Concat}
         {\Gamma(x_1) = \KString_U \\ \Gamma(x_2) = \KString_{U'}}
         {\FType{\Gamma}{x_1 + x_2}{\KString_{UU'}}{\{\monneutral\}}}
  \end{mathpar}

  \caption{The \FJEUS declarative type system}
  \label{fig:type-system-declarative-complete}
\end{figure}

\subsection{Soundness of the Declarative Type System}\label{app:proof}

We now give a formal interpretation of the typing judgment in form of
a soundness theorem. 
Namely, we prove soundness of declarative type system with respect to
operational semantics through interpretation of typing judgment using
heap typings.  Specifically, our interpretation of the typing judgment
$\FType{\Gamma}{e}{\tau}{U}$ states that whenever a well-typed program
is executed on a heap that is well-typed with respect to some heap
typing $\Pi$, then the final heap after the execution is well-typed
with respect to some possibly larger heap typing $\Pi'$.

A \emph{heap typing} $\Pi : \Locs \FnToFin \ATypes$ 
assigns to each heap location a static class (an upper
bound of the actual class found at that location) and a region for
ordinary objects.

We define a typing judgment for values $\WellFormed \Pi v \tau$,
which means that according to heap typing $\Pi$, the value $v$ may
be typed with $\tau$. In particular, the information in $\Pi(l)$
specifies the type of $l$.  
\begin{mathpar}
  \Rule{}
       { }
       {\WellFormed \Pi \NullVal \NullType_\emptyset}
  \qquad
  \Rule{}
       {\Pi(l)=C_r}
       {\WellFormed \Pi l {C_{r}}}
  \qquad       
  \Rule{}
       {\WellFormed \Pi v \sigma \\ \sigma \subtypeOf \tau}
       {\WellFormed \Pi v \tau}
\end{mathpar}
Also, the typing judgment of locations is lifted to stores and
variable contexts:
\begin{eqnarray*}
  \WellFormed \Pi s \Gamma & \text{ iff }  & 
     \WellFormed \Pi {s(x)}{\Gamma(x)}, \text{ for all } x\in\domof(\Gamma)
\end{eqnarray*}

A heap $h$ is \emph{well-typed} with respect to a heap typing $\Pi$
and implicitly a declarative class table $(\TPool_\mathsf{d},F,M_\mathsf{d})$, written
$\SemHeap{\Pi}{h}$, if the following conditions hold: 
\begin{enumerate}
\item at each location, non-string objects are only typed with types
  in the relevant type set~$\TPool_\mathsf{d}$: for any $l\in\dom(\TPool_\mathsf{d})$,
  if $\Pi(l)=C_r$ and $C\neq\KString$, then $C_{\set{r}}\in\TPool_\mathsf{d}$, 
\item the object at each location is ``valid'' with respect to the
  type predicted by $\Pi$ for that location:
  $l\in\domof(h) \text{ and } \SemHeapVal {\Pi} {h(l)} {\Pi(l)},
  \text{ for all $l\in\domof(\Pi)$}$, where
  \begin{eqnarray*}
    \SemHeapVal{\Pi}{(C,F')}{D_r} & \text{ iff } 
    & C = D, \domof(F') = \fList(C), \text{ and } 
    \\ 
    & & \WellFormed \Pi {F'(f)} {F(f,C_r)}, \text{ for all $f\in \fList(C)$},
    \\
    \SemHeapVal{\Pi}{(\Str,w)}{\KString_u} & \text{ iff } & [w] = u.
\end{eqnarray*}
\end{enumerate}
Note that the condition $C = D$ is logically equivalent to the
requirement that if $\Pi(l) = C_r$ then $\classOfVal_h(l) = C$.

The following lemma follows from the well-typedness definition of
heaps.

\begin{lemma}\label{lem:downcast}
  For any heap~$h$, heap typing~$\Pi$, value~$v$, class~$C$, and
  region set~$R$, if $\SemHeap{\Pi}{h}$ and $\WellFormed{\Pi}{v}{C_R}$,
  then $\WellFormed{\Pi}{v}{D_R}$, where $D=\classOfVal_h(v)$.
\end{lemma}

\begin{proof}
  The case $v=\NullVal$ is trivial, so assume $v\in\Locs$.
  Let $\Pi(v) = E_r$ for some class $E$ and region $r$. 
  From $\WellFormed{\Pi}{v}{C_R}$ we obtain, by rule
  inversion, that $E_r \subtypeOf C_R$. That is, $E\preceq C$ and $r\in R$.
  From $\SemHeap{\Pi}{h}$, we obtain that $D = E$. Thus
  $\WellFormed{\Pi}{v}{D_r}$ by definition and
  $\WellFormed{\Pi}{v}{D_R}$ by subsumption.
\end{proof}

As the memory locations are determined at runtime, the heap typings
cannot be derived statically.  Instead, our interpretation of the
typing judgment $\FType{\Gamma}{e}{\tau}{U}$ states that whenever a
well-typed program is executed on a heap that is well-typed with
respect to some typing $\Pi$, then the final heap after the execution
is well-typed with respect to some possibly larger heap typing $\Pi'$.
The typing $\Pi'$ may be larger to account for new objects that may
have been allocated during execution, but the type of locations that
already existed in $\Pi$ may not change.  More formally, a heap typing
$\Pi'$ \emph{extends} a heap typing $\Pi$, written
$\HTextends{\Pi'}{\Pi}$, if $\domof(\Pi)\subseteq\domof(\Pi')$ and
$\Pi(l)=\Pi'(l)$, for all $l\in\domof(\Pi)$.

\begin{theorem}[Soundness Theorem]\label{thm:decl-soun}
  Let $\Prg$ be a well-typed program. For all $\Pi$, $\Gamma$, $\tau$,
  $s$, $h$, $e$, $v$, $h'$, $w$ with
  \begin{displaymath}
      \FType{\Gamma}{e}{\tau}{U} \quad\text{and}\quad
      \WellTyped \Pi s \Gamma \quad\text{and}\quad
      \SemHeap{\Pi}{h} \quad\text{and}\quad
      \FSem{s}{h}{e}{v}{h'}{w}\\
  \end{displaymath}
  there exists some $\HTextends {\Pi'} \Pi$ such that
  \begin{displaymath}
    \WellTyped {\Pi'} v \tau \quad\text{and}\quad 
    \SemHeap{\Pi'}{h'}
    \quad\text{and}\quad [w]\in U.
  \end{displaymath}
\end{theorem}

Note that Theorem~\ref{thm:decl-sound-shortened} is a
corollary of this theorem.

\begin{proof} 
  In~\cite{comlan} a version of this theorem without effects and
  builtin functions is presented. The result here follows similarly
  by induction over the sum of the depth of the derivation of the
  operational semantics judgment and the depth depth of the derivation
  of the typing judgment.
  %

  First, let us consider the case where $\FType{\Gamma}{e}{\tau}{U}$
  has been derived by subtyping rule. By rule inversion we get
  $\FType{\Gamma}{e}{\tau'}{U'}$ where $\tau'\subtypeOf \tau$ and
  $U' \subseteq U$. By induction we 
  obtain $\WellTyped {\Pi'} v \tau'$ and $\SemHeap{\Pi'}{h'}$
  and $[w]\in U'$ for some $\HTextends {\Pi'} \Pi$. Since we
  have $\tau'\subtypeOf \tau$ and $U' \subseteq U$, we deduce that 
  $\WellTyped {\Pi'} v \tau'$ and $[w]\in U'$.  

  Next, we assume that the subtyping rule is not the last used rule of
  the typing judgment.
  We perform a case distinction over the last rule used in the
  derivation of operational semantics judgment.
  In the rest of the proof, we only consider the cases of object
  allocations, type casts, and (internal or external) method
  invocations, the other cases being similar and simpler.
  We assume that $P$ is well-typed w.r.t. the declarative class table
  $(\TPool_\mathsf{d},F,M_\mathsf{d})$.

  \begin{itemize}

  \item $\FSem{s}{h}{\New C}{l}{h[l\mapsto (C,F)]}{\varepsilon}$.

    \smallskip

    The typing judgement is
    $\FType{\Gamma}{\New{C}}{C_{\set{r}}}{\set{\monneutral}}$ and thus
    $C_{\set{r}}\in\TPool_\mathsf{d}$. Also, from semantic rule we have
    $h'=h[l\mapsto(C,F)]$ where $l\not\in\domof(h)$ and $F =
    [f\mapsto\NullVal]_{f\in\mathit{fields}(C)}$. Since
    $\domof(\Pi)\subseteq\domof(h')$, we have
    $l\not\in\domof(\Pi)$. We choose $\Pi'=\Pi[l\mapsto (C,r)]$, and
    thus have $\WellFormed \Pi l {C_r}$.  By definition we also have
    $\classOfVal_h(l) = C$.  To show $\SemHeap{\Pi'}{h}$, it suffices
    to show $\SemHeapVal {\Pi'} {h(l)} {C_r}$.  This follows
    trivially, as $h'(l)=(C,F)$ and $C=C$ and $F(f)=\NullVal$ for all
    $f\in\domof(F)$.

    \smallskip

  \item $\FSem{s}{h}{x.m(\seq y)}{v}{h'}{w}$.

    \smallskip

    By rule inversion, we know there is a location $l\in\domof(h)$
    such that $s(x)=l$ and $h(l)=(D,\_)$, and
    $\FSem{s'}{h}{\mathit{mtable}(D,m)}{v}{h'}{w}$ where
    $s'=[this\mapsto l] \cup [x^m_i\mapsto
      s(y_i)]_{i\in\set{1,\ldots,\ar(m)}}$.
    
    \smallskip

    The typing judgment is $\FType{\Gamma,x:C_R,\seq
      y:\seq{\sigma}}{x.m(\seq y)}{\tau}{U}$.  As $\WellFormed \Pi s
    {(\Gamma,x:C_R,\seq y:\seq\sigma)}$, we have $\WellFormed \Pi l
    {C_R}$. By definition of well-typed values w.r.t.~$\Pi$, there
    must exist some class $E$ and some region $r$ with $\Pi(l)=E_r$
    such that $r\in R$ and $E\subclassOf C$.  We can derive
    $\WellFormed \Pi l {C_{\{r\}}}$.  With $\SemHeap{\Pi}{h}$, we
    infer that $D=E$ and thus $D\subclassOf C$.  Also, as
    $\SemHeap{\Pi}{h}$, we have that $E_r\in \TPool_\mathsf{d}$, that is,
    $D_r\in \TPool_\mathsf{d}$.

    \smallskip

    By rule \rulename{TD-Invoke}, we know there exists a method typing
    $\MethodTP{\seq\sigma'}{\tau'}{U'} \in M_\mathsf{d}(m, C_r)$ such that
    $(\bar\sigma', \tau',U') \msqsubseteq (\bar\sigma, \tau, U)$.  As
    $D\subclassOf C$, $D_r\in \TPool_\mathsf{d}$, and the class table is
    well-formed, there is a method typing
    $\MethodTP{\seq\sigma''}{\tau''}{U} \in M_\mathsf{d}(m, D_r)$ such that
    $(\bar\sigma'', \tau'',U'') \msqsubseteq (\bar\sigma', \tau', U')$.
    From $\WellFormed \Pi s {(\Gamma,x:C_R,\seq y:\seq\sigma)}$, it
    follows $\WellFormed \Pi {s(y_i)} {\sigma_i''}$ for
    $i\in\{1,\ldots,\ar(m)\}$.

    \smallskip

    As $\Prg$ is well-typed, we get
    $\FType{\Gamma'}{\mathit{mtable}(D,m)}{\tau''}{U}$ where
    $\Gamma'=[this\mapsto C_{\{r\}}] \cup
    [x^m_i\mapsto\sigma_i'']_{i\in\{1,\ldots,\ar(m)\}}$. From the
    facts from above, we get $\WellFormed \Pi {s'}{\Gamma'}$, so we
    can finally apply the theorem inductively on the derivation of the
    semantics and get $[w]\in U$, and $\WellFormed {\Pi'} v {\tau''}$
    and $\SemHeap{\Pi'}{h'}$ for some $\HTextends{\Pi'} \Pi$. From
    $\tau''\subtypeOf\tau$ we obtain $\WellFormed {\Pi'} v \tau$.

    \smallskip

  \item ${\FSem{s}{h}{\fn(\seq y)}{l}{h'}{w'}}$, where $\fn$ is an
    external method.

    \smallskip

    By rule inversion, there is a location $l \notin \domof(h)$, a
    string literal $\Str\in\Strings$ and a words $w\in\Sigma^*$ such
    that $h' = h[l \mapsto (\Str,w)]$ and $(\Str, w, w') \in
    \sem(\fn)((\Str_1,w_1), \ldots, (\Str_n,w_n))$, where $n =
    \ar(\fn)$ and $h(s(y_i)) = (\Str_i, w_i)$, for each $i$ with $i
    \in \{1, \ldots, n\}$.

    \smallskip

    From $\WellFormed \Pi s \Gamma$, it follows $\WellFormed \Pi
    {s(y_i)} {\KString_{R_i}}$.
    By definition of well-typed values w.r.t.~$\Pi$, there must
    exist some region $r_i$ with $\Pi(s(y_i))=\KString_{r_i}$ and
    $r_i\in R_i$.
    Furthermore, as $\SemHeap{\Pi}{h}$, we obtain that $[w_i]=r_i$.
    Then, from the properties of the typing function~$M$ (see
    Section~\ref{sec:policy}), we get that $[w]\in R$ and $[w']\in U$,
    where $(R,U) = M(\fn)(r_1, \dots, r_n)$.

    \smallskip

    We let $\Pi'=\Pi[l\mapsto (\KString,[w])]$. Then, as
    $\SemHeap{\Pi}{h}$, we directly get that ${\SemHeap{\Pi'}{h'}}$.
    Assume that the typing judgment is ${\FType{\Gamma}{\fn(\seq
        y)}{\KString_{R'}}{U'}}$ for some $R'\subseteq\Regions$ and
    $U'\in\Effects$.
    Then, by the rule \rulename{TD-Builtin}, we have
    $M(\fn)(r_1,\dots,r_n) \sqsubseteq (R', U')$.
    That is, $(R,U) \sqsubseteq (R', U')$. Said otherwise
    $R\subseteq R'$ and $U \subseteq U'$.
    It then follows that $[w]\in R'$ and $[w']\in U'$. The former
    conjunct shows that $\WellTyped{\Pi'}{l}{\KString_{R'}}$.

    \smallskip

  \item ${\FSem{s}{h}{\Cast{e}{C}}{v}{h}{w}}$.

    The typing judgment is
    $\FType{\Gamma}{\Cast{e}{C}}{C_R}{U}$. Using rule inversion we get
    that $\FType{\Gamma}{e}{D_R}{U}$.
    By the induction hypothesis, there is $\HTextends{\Pi'}{\Pi}$ such
    that $\WellTyped{\Pi'}{v}{D_R}$, $\SemHeap{\Pi'}{h}$, $[w]\in U$.
    By applying Lemma~\ref{lem:downcast} we get
    $\WellTyped{\Pi'}{v}{E_R}$, where $E = \classOfVal_h(v)$. As
    $\classOfVal_h(v) \preceq C$ from the premise of the operational
    semantics rule that was used for this case, we obtain
    $\WellTyped{\Pi'}{v}{C_R}$ by subsumption.
  \end{itemize}
\end{proof}

\subsection{Algorithmic Type and Effect System}
\label{app:alg-system}

Figure~\ref{fig:type-system-algorithmic-complete} lists the complete
rules of the algorithmic type system.
The rules
are in essence more specialized versions of the ones in the
semi-declarative system. Notice that in the rules \rulename{TA-If},
\rulename{TA-Invoke} and \rulename{TA-GetF} we compute the least upper
bound, i.e. join, of two or more types in premises.

\begin{figure}
  \begin{mathpar}
    \Rule{TA-Var}
         { }
         {\AType{\Gamma,x:\tau}{z}{x}{\tau}{\{\monneutral\}}}
    \and
    \Rule{TA-If}
         {\AType{\Gamma, x: C_{R \cap S}, y : D_{R \cap S}}{z}{e_1}{\tau}{U} \\ 
          \AType{\Gamma, x: C_{R}, y : D_{S}}{z}{e_2}{\tau'}{U'}\\
        }
         {\AType{\Gamma}{z}{\IfEqual{x}{y}{e_1}{e_2}}{\tau \sqcup \tau'}{U \cup U'}}
    \and
    \Rule{TA-Let}
         {\AType{\Gamma}{z}{e_1}{\tau}{U} \\
           \AType{\Gamma,x:\tau}{z}{e_2}{\tau'}{U'}}
         {\AType{\Gamma}{z}{\Let{x}{e_1}{e_2}}{\tau'}{UU'}}
    \and
    \Rule{TA-Null}
         { }
         {\AType{\Gamma}{z}{\Null}{\NullType_{\emptyset}}{\{\monneutral\}}}
    \and
    \Rule{TA-New}
         {r=\psi(z,i) \\ C_{r}\in\TPool_\mathsf{a}(z)}
         {\AType{\Gamma}{z}{\LExpr{\New C}{i}}{C_{\set{r}}}{\set{\monneutral}}}
    \and
    \Rule{TA-Cast}
         {\AType{\Gamma}{z}{e}{C_R}{U} \\ C\subtypeOf D \text{ or } D\subtypeOf C}
         {\AType{\Gamma}{z}{\Cast{e}{D}}{D_R}{U}}
    \and
    \Rule{TA-Invoke}
       { 
         (\tau, U) = \bsqcup\setx{M_\mathsf{a}(m, z', C_r, \bar{\sigma})}{r\in R, z'=\phi(z,C,r,m,i)}}
       {\AType{\Gamma,x:C_R,\seq y:\seq{\sigma}}{z}{\LExpr{x.m(\seq y)}{i}}{\tau}{U}}
    \and
    \Rule{TA-Builtin}
         {\ar(\fn)=n \quad \text{$\Gamma(x_1) = \KString_{R_1}, \ldots, 
             \Gamma(x_n) = \KString_{R_n}$} \\\\  
           (\tau,U) = \bsqcup\setx{M(\fn) (r_1, \ldots, r_n)}{r_1 \in R_1, \ldots, r_n \in R_n}}
         {\AType{\Gamma}{z}{\fn(x_1, \dots, x_n)}{\tau}{U}}
    \and
    \Rule{TA-GetF}
         {\tau = \bsqcup\setx{F(f,C_r)}{r \in R}}
         {\AType{\Gamma,x:C_R}{z}{x.f}{\tau}{\{\monneutral\}}}
    \and
    \Rule{TA-SetF}
         {\tau \subtypeOf F(f,C_r),\ \text{for all $r\in R$}}
         {\AType{\Gamma,x:C,y:\tau}{z}{\Assign{x.f}{y}}{\tau}{\{\monneutral\}}}
    \and
    \Rule{TA-Lit}
         {\wordoflit(\Str) = w}
         {\AType{\Gamma}{z}{"\Str"}{\KString_{\{[w]\}}}{\{\monneutral\}}}
    \and
    \Rule{TA-Concat}
         {\Gamma(x_1) = \KString_U \\ \Gamma(x_2) = \KString_{U'}}
         {\AType{\Gamma}{z}{x_1 + x_2}{\KString_{UU'}}{\{\monneutral\}}}
  \end{mathpar}

  \caption{The \FJEUS algorithmic type system}
  \label{fig:type-system-algorithmic-complete}
\end{figure}

\subsubsection{Proof of Completeness}

\begin{proof}[of Theorem~\ref{thm:algocomplete}]
Given a semi-declarative method typing $M_\mathsf{s}$ we define $\Phi(M_\mathsf{s})$ as 
the semi-declarative method typing obtained by 
``best'' typing of all the methods according to $M_\mathsf{s}$.
That is, $\Phi(M_\mathsf{s})$ returns the most precise type (i.e. the smallest type 
w.r.t.~$\msqsubseteq$ in Section~\ref{sec:ctables_semi}) for each 
method $m$ in $\Prg$. 
Well-typedness w.r.t. $(\TPool_\mathsf{s},F,M_\mathsf{s})$ essentially says that 
$\Phi(M_\mathsf{s}) \sqsubseteq M_\mathsf{s}$.
Accordingly, if $M_{\infty}$ is the least fixpoint of $\Phi$ then 
$M_{\infty} \sqsubseteq M_\mathsf{s}$.

Now, $M_{\infty}$ is obtained as the limit of the chain $\bot := M_{s_0} 
\sqsubseteq \Phi(\bot):=M_{s_1} \sqsubseteq  M_{s_2} \ldots$  
where $M_{s_{(i+1)}} = \Phi(M_{s_i})$.

Induction on $i$ shows that each typing $M_{s_i}$ is 
actually a function, i.e.~for each method $m$, each context $z$,
each class~$C$, each region $r$, and each $\bar\sigma$ there 
is \emph{exactly one} region type and effect $(\tau, U)$ such that 
$(\bar\sigma,\tau,U) = M_{s_i}(m,z,C_r).$
Now, using the definition of $\atoms(\bar\sigma)$ we can 
construct $M_\mathsf{a}$ as $(\tau, U) = M_\mathsf{a}(m,z,C_r, \bar\sigma_\mathsf{a})$.  
Therefore, we can restrict attention to algorithmic method tables 
throughout. 

Regarding well-typedness, from the construction of $M_\mathsf{a}$ 
it follows that if $P$ is well-typed w.r.t. $(\TPool_\mathsf{s}, F, M_\mathsf{s})$ then 
it is also well-typed w.r.t $(\TPool_\mathsf{s}, F, M_\mathsf{a})$. \qed
\end{proof}

\section{Details of the Type Inference Algorithm}
\label{app:alg}

From the algorithmic type system in Section~\ref{sec:algo}, a type inference
algorithm can easily be constructed by reading them as a functional program.
This algorithm computes the type
and effects of any expression~$e$, given a typing context~$\Gamma$, a
method call context~$z$, a class table $(\TPool_\mathsf{a},F_\mathsf{a},M_\mathsf{a})$, and the
implicit type systems' parameters. We defer the rather obvious details to Appendix~\ref{app:pseudocode} and henceforth, assume the existence of
a procedure~$\te$, which, when called as
$\te(e,\Gamma,z,(\TPool_\mathsf{a},F_\mathsf{a},M_\mathsf{a}))$, returns a tuple~$(\tau,U)$.

We present next a more general type inference algorithm,
denoted~$\alg$, that infers an algorithmic class table for a given
program~$P$, provided the standard Java types of the program's methods
and fields are also given in form of a class table~$(F_0,M_0)$.
As output it returns an (algorithmic) class table~$(\TPool_\mathsf{a},F_\mathsf{a},M_\mathsf{a})$
such that $P$ is well-typed with respect to it.
Thus the algorithm can be readily used to check whether an
expression~$e$ follows a guideline. Indeed, it is sufficient to call
$\te$ on the expression~$e$ for the newly computed class table and
check whether $U\subseteq\Allowed$, where $(\tau,U)$ is the tuple
returned by the call to $\te$.

\begin{listing}[htb]
\begin{lstlisting}[language=myML]
proc $\alg(P, (F_0,M_0))$
  $(\TPool_\mathsf{a},F_\mathsf{a}, M_\mathsf{a})$ := lift($P$, $(F_0, M_0$))
  do  
    $(F', M')$ := $(\TPool',F_\mathsf{a}, M_\mathsf{a})$ 
    foreach $(m,z,C_r,\seq{\sigma}) \in \domof(M_\mathsf{a})$ with $(C,m)\in\dom(\mTable)$
      $e$ := $\mTable(C,m)$
      $\Gamma$ := $[\this\mapsto C_{r}]\cup[ x^m_i\mapsto\sigma_i]_{i\in{\{1,\ldots,\ar(m)\}}}$
      $(F_\mathsf{a},\tau,U)$ := $\tep(e, \Gamma, z, (F_\mathsf{a}, M_\mathsf{a}))$
      $M_\mathsf{a}(m,z,C_r,\seq{\sigma})$ := $M_\mathsf{a}(m,z,C_r,\seq{\sigma})\sqcup (\tau,U)$
      $(F_\mathsf{a}, M_\mathsf{a})$ := checkClassTable($F_\mathsf{a}$, $M_\mathsf{a}$)
  until $(F',M') = (F_\mathsf{a},M_\mathsf{a})$
  return $(F_\mathsf{a},M_\mathsf{a})$
\end{lstlisting}
\caption{The interprocedural analysis.}
\label{code:inter}
\end{listing}

The pseudo-code of the algorithm~$\alg$ is given in
Listing~\ref{code:inter}.
The algorithm computes the field and method typings $(F_\mathsf{a},M_\mathsf{a})$
iteratively, by a fix-point computation.  The tuple $(F_\mathsf{a},M_\mathsf{a})$ is
initialized with the standard class table, by lifting it to an
algorithmic class table. More precisely, initially all entries in the
tables $F_\mathsf{a}$ and $M_\mathsf{a}$ are set to the lowest possible elements
of~$\Lattice$ (w.r.t.~$\sqsubseteq$). For instance,
$F_\mathsf{a}(f,C_r)=(D_\emptyset,\set{\monneutral})$, where $D = M_0(C,f)$,
for any class~$C$ and field $f\in\fList(C)$. Also, the set~$\TPool_\mathsf{a}$
of relevant refined types is simply obtained by inspecting all
allocations in $P$ in all contexts and adding the corresponding
refined types to $\TPool_\mathsf{a}$, and then taking the closure of this set
by "supertyping".

Next, the fix-point is iteratively computed with a \ls{do-until}
loop. Intuitively, in each iteration, all method bodies of the program
are checked against the relevant entries of the method typing~$M_\mathsf{a}$
and if the check fails than the corresponding entries are updated and
their types ``weakened''. This is performed in the pseudo-code as
follows. For each method typing entry, the type and effect of the
corresponding method body is computed using a variant $\tep$ of the
type inference algorithm~$\te$.

The variant~$\tep$ behaves exactly as $\te$, except it also possibly
updates the table~$F_\mathsf{a}$. This may happen for field write
sub-expressions $\Assign{x.f}{y}$. Recall that the corresponding type
rule requires that~$y$ must have a subtype of the types allowed for
the field~$f$ by the $F_\mathsf{a}$ table with respect to each possible region
$r \in R$, where $\Gamma(x)=C_R$. If this condition is not satisfied
then $F_\mathsf{a}$ is updated such that it is satisfied, by weakening the
offending entry with the type of~$y$. Formally, if
$\Gamma(y)\not\subtypeOf F_\mathsf{a}(f,C_r)$ for some $r\in R$, then
$F_\mathsf{a}(f,C_r)$ is set to $F_\mathsf{a}(f,C_r) \sqcup \Gamma(y)$.

Once the type and effect $(\tau,U)$ of the analyzed method body is
computed, then the current entry of $M_\mathsf{a}$ is updated if
needed. Finally, the \ls{checkClassTable} procedure checks whether the
tuple $(F_\mathsf{a},M_\mathsf{a})$ satisfies the constraints from the definition of
algorithmic class tables and if not it updates the offending entries
so that the constraints are satisfied. For instance, if $C\preceq D$
and $F_\mathsf{a}(f,C_r) \neq F_\mathsf{a}(f,D_r)$ for some field $f$ and region $r$,
then both $F_\mathsf{a}(f,C_r)$ and $F_\mathsf{a}(f,D_r)$ are set to
$F_\mathsf{a}(f,C_r) \sqcup F_\mathsf{a}(f,D_r)$.

The pseudo-code of the auxiliary procedures \ls{lift},
\ls{checkClassTable}, and $\tep$ can be found later in this section.

\begin{theorem}
  Let $P$ be a program that is well-typed with respect to some class
  table $(F_0,M_0)$. The procedure $\alg(P,(F_0,M_0))$ terminates
  and returns an algorithmic class table $(F_\mathsf{a},M_\mathsf{a})$ such that $P$ is
  well-typed with respect to $(F_\mathsf{a},M_\mathsf{a})$.
\end{theorem}
\begin{proof}
  For termination, we note that in each iteration of the \ls{do-until}
  loop at least one of the tables $F_\mathsf{a}$ and $M_\mathsf{a}$
  changes. Furthermore, each change can only ``worsen'' the type of
  the changed table entry. Finally, each entry can only be changed a
  finite number of types, as the last possible change leads to the
  tuple $(\KObject_\Regions,\Powerset{\Monoid})$, which cannot be
  ``worsened''.

  The correctness of the algorithm relies on the correctness of the
  $\tep$ procedure, which is ensured by
  Theorem~\ref{thm:algo-sound}. Finally, the \ls{checkClassTable}
  procedure ensures that the tuple $(F_\mathsf{a},M_\mathsf{a})$ satisfies the
  constrains from the definition of algorithmic class tables.
\end{proof}

\subsection{Pseudocode of the \textsf{typeff} Procedure}
\label{app:pseudocode}

The whole algorithm is presented as imperative style (with pattern matching)
pseudo-code in Listings~\ref{code:inter}, \ref{code:aux},
and~\ref{code:intra}.
In the pseudo-code, the symbol \ls{:=} denotes a pattern match
followed by an assignment. For instance, if at a point in the
pseudo-code $x$ is a fresh variable, while $y$ is already used, then
\ls{$(x,y)$ := $(a,b)$} first checks that $y$ equals $b$ and then
assigns $x$ to $a$. %
If the check fails, then an exception is raised.
The pseudocode of the procedure \ls{checkMethodTyping} is omitted,
being very similar to that of the \ls{checkFieldTyping} procedure.

\begin{listing}
\begin{lstlisting}[language=myML]
proc $\tep$($e$, $\Gamma$, $z$)
  match $e$ with
  $\mid x$ -> return $(\Gamma(x), \monneutral)$
  $\mid \Let{x}{e}{e'}$ ->
     $\tau, U$ := $\tep$($e$, $\Gamma$, $z$)
     $\tau', U'$ := $\tep$($e'$, $\Gamma[x\mapsto \tau]$, $z$)
     return $(\tau', UU')$
  $\mid \IfEqual{x}{y}{e_1}{e_2}$ ->
     return $\tep$($e_1$, $\Gamma$, $z$) $\sqcup$ $\tep$($e_2$, $\Gamma$, $z$)
  $\mid \Null$ -> return $(\NullType_\emptyset,\monneutral)$
  $\mid \LExpr{\New C}{i}$ -> return $(C_{\set{\psi(z,i)}}, \monneutral)$
  $\mid \Cast{e}{D}$ -> 
     $(C_R, U)$ := $\tep$($e$, $\Gamma$, $z$)
     if $C\preceq D$ or $D\preceq C$ then
       return $(D_R, U)$
     else raise "type error"
  $\mid x.f$ -> 
     $C_R$ := $\Gamma(x)$
     return $(\bsqcup \setx{F_\mathsf{a}(f,C_r)}{r\in R}, \monneutral)$
   $\mid \Assign{x.f}{y}$ ->
     $C_R$ := $\Gamma(x)$
     foreach $r\in R$ with $\Gamma(y)\not\subtypeOf F_\mathsf{a}(f,C_r)$
       $F_\mathsf{a}(f,C_r)$ := $F_\mathsf{a}(f,C_r) \sqcup \Gamma(y)$
     return $(\Gamma(y),\monneutral)$
  $\mid \LExpr{x.m(\seq y)}{i}$ ->
     $C_R$ := $\Gamma(x)$
     return $\bsqcup \setx{M_\mathsf{a}(m,z',C_r,\Gamma(\bar{y}))}{r\in R, z' = \phi(z,C,r,m,i)}$
  $\mid \fn(\seq y)$ ->
     foreach $i\in\set{1,\dots,|\seq y|}$
       $\KString_{R_i}$ := $\Gamma(y_i)$
     return $\bsqcup\setx{M(\fn) (r_1, \ldots, r_n)}{r_1 \in R_1, \ldots, r_n \in R_n}$
  $\mid "\Str"$ -> return $(\KString_{[\wordoflit(\Str)]}, \monneutral)$
  $\mid x + x'$ ->
     $\KString_U$ := $\Gamma(x)$
     $\KString_{U'}$ := $\Gamma(x')$
     return $(\KString_{UU'}, \monneutral)$
\end{lstlisting}
\caption{The type inference algorithm.}
\label{code:intra}
\end{listing}

\begin{listing}[htb]
\begin{lstlisting}[language=myML]
proc lift($P$, $(F_0, M_0)$)
  foreach $C\in \Classes(P)$
    foreach $m\in\mList(C)$
      $((D_1, \dots, D_{\ar(m)}), E)$ := $M_0(C, m)$
      foreach $r\in\Regions, z\in\Contexts, \bar{s}\in \Regions^{\ar(m)}$
        $\bar{\sigma}$ := $\big((D_1,s_1),\dots,(D_{\ar(m)},s_{\ar(m)})\big)$
        $M_\mathsf{a}(m,z,C_r,\bar{\sigma})$ := $(E_\emptyset, \set{\monneutral})$
    foreach $f\in\fList(C)$
      $D$ := $F_0(C, f)$
      foreach $r\in\Regions$
        $F_\mathsf{a}(f,C_r)$ := $D_\emptyset$
  $\TPool_\mathsf{a}$ := $\emptyset$
  foreach $\LExpr{\New{C}}{i}\in \mathsf{expr}(P)$, $z\in\Contexts$
    $\TPool_\mathsf{a}$ := $\TPool_\mathsf{a}\cup \set{C_{\set{\psi(z,i)}}}$
  $\TPool_\mathsf{a}$ := closure($\TPool_\mathsf{a}$)
  return $(\TPool_\mathsf{a},F_\mathsf{a},M_\mathsf{a})$

proc checkClassTable($F_\mathsf{a}$, $M_\mathsf{a}$)
  checkFieldTyping()
  checkMethodTyping()
  return $(F_\mathsf{a}, M_\mathsf{a})$

proc checkFieldTyping()
  updateF($\KObject$)

proc updateF($C$)
  foreach $(C',r,f)\in\domof(F_\mathsf{a})$, $D\prec C$ with $C=C'$
    makeEqual($C$, $D$, $r$, $f$)
    updateF($D$)
    makeEqual($C$, $D$, $r$, $f$)

proc makeEqual($C$, $D$, $r$, $f$)
  if $F_\mathsf{a}(f,C_r) \neq F_\mathsf{a}(f,D_r)$ then
    $F_\mathsf{a}(f,C_r) := F_\mathsf{a}(f,C_r) \sqcup F_\mathsf{a}(f,D_r)$
    $F_\mathsf{a}(f,D_r) := F_\mathsf{a}(f,C_r)$
\end{lstlisting}
\caption{Auxiliary procedures for the interprocedural analysis.}
\label{code:aux}
\end{listing}

Note that in an implementation the tables $M_\mathsf{a}$ and $F_\mathsf{a}$ would be
built on-the-fly. That is, they are not initialized for all possible
regions and contexts, but rather an entry is added to the table only
when the corresponding value is not the default one. For instance,
$F_\mathsf{a}(f,C_r)$ only contains entries $(E_R,U)$ with $E\neq D$, 
$R\neq\emptyset$, or $U\neq\set{\monneutral}$, where $D=M_0(C,f)$.

\section{Additional Details on the Implementation}\label{app:eval}

We note that \FJEUS does not feature loops, assignments, and
arrays. Instead, loops should be encoded through recursive methods. 
In the implementation, we use Soot's intraprocedural analysis to
handle loops.  It is not hard to see that this gives the exact same
result as if we would have introduced recursively defined methods for
loops.
Also, we treat assignments of local variables (but not of fields and
array elements) as strong updates, that is, when analyzing an
assignment $\Assign{x}{e}$, the type of $x$ becomes $\tau'$, where
$\tau'$ is the type of $e$.
Finally, arrays are treated by extending refined types $C_R$ to types
$C[\,]_R$, with the expected rules: an array access~$a[i]$ has type
$C_R$ if $a$ has type~$C_R$, and, following an array update
$a[i] := e$, the type of $a$ becomes $D[\,]_{R'}$ with
$D_{R'} = C_R\sqcup\tau'$, where $\tau'$ is the type of $e$. That is,
all array elements have the same refined type. Our mock code for Java
collections results in a similar behavior.

The implementation does not support reflection, two-dimensional
arrays, concurrent features, and character-level string
operations. This means that no guarantees are provided for programs
that use the non-supported features.  However, our analysis is sound
on the fragment we can handle, though not complete. That is, we may
obtain false positives, but never obtain false negatives.

The implementation uses the following type system parameters.  
We let $\Contexts=\bigcup_{i\in\{0,\dots,k\}}\Sites^i$ and
$\Regions=\Contexts\times\Sites$, where $k\geq 0$ is a parameter of
the analysis.
We take $\psi(z,i)=(z,i)$ and $\phi(z,C,r,m,i) = (i :: z)|_{k}$,
where $L|_{k}$ is the truncation of the list $L$ after the first $k$
elements.
Therefore the context of a method call is the call string (i.e. an
abstraction of the call stack) restricted to the last $k$ calls; this
corresponds the common $k$-CFA abstraction principle. Objects are
distinguished their allocation site and by the current context.


\section{Additional Examples}

\subsection{FJ\textsubscript{sec} Example}
\label{app:fjsec}

The following example illustrates that having the set
$\TPool_\mathsf{d}\subsetneq\ATypes$ improve the precision of the analysis. The
example is inspired from~\cite{DBLP:journals/lisp/Skalka08}. There, to
treat the corresponding example, ``soft subtyping'' has been used.

\begin{lstlisting}[language=java]
class Input {
  String get() { return "nontainted string";	} }

class InputExt extends Input {
  String get() { return getString(); } }

public class FJsec {	
  void m_ok() {
    Input i = new Input();
    String s = i.get();
    putString(s);
  }
  
  void m_bad() {
    Input i = new InputExt();
    String s = i.get();
    putString(s);
  }
}
\end{lstlisting}

\subsection{The SecuriBench test case Aliasing3}
\label{app:aliasing3}

The following code constitutes the Aliasing3 test case from the
SecuriBench benchmark, with minor editing for readability.

\begin{lstlisting}[language=java]
void doGet(HttpServletRequest req, HttpServletResponse resp) throws IOException {
  String name = req.getParameter("name");
  String[] a = new String[10];
  String str = a[5];
  a[5] = name;
  name = str;
              
  PrintWriter writer = resp.getWriter();
  writer.println(str);                        /* BAD */
}
\end{lstlisting}

We believe that the test case is wrongly marked as
problematic. Indeed, the value of \ls{str} when given to the
\ls{println} method is $\NullVal$ and this does not constitute a
policy violation.

\subsection{securibench-micro.Strong\_updates3}
\label{app:strongupdates3}

\begin{lstlisting}[language=java]
    class Widget {
        String value = null;
    }
    protected void doGet(HttpServletRequest req, HttpServletResponse resp) throws IOException {
        String name = req.getParameter(FIELD_NAME);
        Widget w = new Widget();
        w.value = name;
        w.value = "abc";

        PrintWriter writer = resp.getWriter();
        writer.println(w.value);              /* OK */
    }
\end{lstlisting}

Here the result of our analysis is a false positive (that is, the tool
reports a potential violation of the guideline), because field typings
are always worsened, never overwritten.

\end{document}